\documentclass{article}
\usepackage{fullpage}
\usepackage[utf8]{inputenc}
\usepackage{amsfonts}

\usepackage{microtype}
\usepackage{amsmath}
\usepackage{amsfonts}
\usepackage{amssymb}
\usepackage{amsthm}
\usepackage{algorithm}
\usepackage{algpseudocode}
\usepackage{color}
\usepackage{xcolor}
\usepackage{enumerate}
\usepackage{url,hyperref}
\usepackage{tikz}
\usetikzlibrary{shapes.geometric}

\usepackage[utf8]{inputenc}
\usepackage{bm}

\usepackage{thmtools}
\usepackage{thm-restate}

\newcommand{\eps}{\varepsilon}

\newcommand{\dd}{\mathrm{d}}

\def\R{{\mathbb R}}

\DeclareMathOperator*{\argmin}{\arg\!\min}
\DeclareMathOperator*{\argmax}{\arg\!\max}
\DeclareMathOperator*{\trace}{\textup{tr}}

\DeclareMathOperator{\cut}{\mathrm{cut}}

\newtheorem{theorem}{Theorem}
\newtheorem{lemma}[theorem]{Lemma}

\newtheorem{definition}[theorem]{Definition}
\newtheorem{conjecture}[theorem]{Conjecture}

\theoremstyle{definition}

\usepackage[T1]{fontenc}

\DeclareMathOperator*{\E}{\mathbb{E}}
\DeclareMathOperator*{\Var}{Var}
\DeclareMathOperator{\erf}{erf}
\DeclareMathOperator{\dom}{dom}

\DeclareMathOperator{\Ind}{Ind}

\newcommand{\N}{\mathbb{N}}
\newcommand{\cS}{\mathcal{S}}
\newcommand{\cP}{\mathcal{P}}

\newcommand{\cG}{\mathcal{G}}
\newcommand{\cT}{\mathcal{T}}
\newcommand{\ones}{\pmb{1}}
\newcommand{\comment}[1]{}

\newcommand{\gReg}[2]{\cG_{#1, #2}^{\textup{reg}}}
\newcommand{\gErd}[2]{\cG_{#1, #2}}

\def\bone{{\bf 1}}

\allowdisplaybreaks

\newcommand{\T}{\mathsf{T}}
\newcommand{\iprod}[1]{\left\langle{#1}\right\rangle}

\title{Cut Sparsification of the Clique Beyond the Ramanujan Bound: A Separation of Cut Versus Spectral Sparsification}
\author{Antares Chen\\University of Chicago
 \and
Jonathan Shi \\ Bocconi University
 \and
Luca Trevisan \\ Bocconi University}

\begin{document}
\maketitle


\begin{abstract}
We prove that a random $d$-regular graph, with high probability, is a cut sparsifier of the clique with approximation error at most $\left(2\sqrt{\frac 2 \pi} + o_{n,d}(1)\right)/\sqrt d$, where $2\sqrt{\frac 2 \pi} = 1.595\ldots$ and $o_{n,d}(1)$ denotes an error term that depends on $n$ and $d$ and goes to zero if we first take the limit $n\rightarrow \infty$ and then the limit $d \rightarrow \infty$.

This is established by analyzing linear-size cuts using techniques of Jagannath and Sen \cite{jagannath2017unbalanced} derived from ideas in statistical physics, and analyzing small cuts via martingale inequalities.

We also prove new lower bounds on spectral sparsification of the clique. If $G$ is a spectral sparsifier of the clique and $G$ has average degree $d$, we prove that the approximation error is at least the ``Ramanujan bound'' $(2-o_{n,d}(1))/\sqrt d$, which is met by $d$-regular Ramanujan graphs, provided that either the weighted adjacency matrix of $G$ is a (multiple of) a doubly stochastic matrix, or that $G$ satisfies a certain high ``odd pseudo-girth'' property. The first case can be seen as an ``Alon-Boppana theorem for symmetric doubly stochastic matrices,'' showing that a symmetric doubly stochastic matrix with $dn$ non-zero entries has a non-trivial eigenvalue of magnitude at least $(2-o_{n,d}(1))/\sqrt d$; the second case generalizes a lower bound of Srivastava and Trevisan~\cite{ST18}, which requires a large girth assumption.

Together, these results imply a separation between spectral sparsification and cut sparsification.  If $G$ is a random $\log n$-regular graph on $n$ vertices, we show that, with high probability (this is to ensure that $G$, and consequently any $d$-regular subgraph, has high pseudogirth), $G$ admits a (weighted subgraph) cut sparsifier of average degree $d$ and approximation error at most $(1.595\ldots + o_{n,d}(1))/\sqrt d$, while every (weighted subgraph) spectral sparsifier of $G$ having average degree $d$ has approximation error at least $(2-o_{n,d}(1))/\sqrt d$.
\end{abstract}


\section{Introduction}

If $G=(V,E_G,w_G)$ is a, possibly weighted,  undirected graph, a {\em cut sparsifier} of $G$ with error $\epsilon$ is a weighted graph $H=(V,E_H,w_H)$ over the same vertex set of $G$ and such that

\begin{equation} \label{eq.cut}
    \forall S\subseteq V \ \ \ (1-\epsilon) \ \cut_G(S) \leq \cut_H(S) \leq (1+\epsilon) \  \cut_G(S)
\end{equation}
where $\cut_G(S)$ denotes the number of edges in $G$ with one endpoint in $S$ and one endpoint in $V-S$, or the total weight of such edges in the case of weighted graphs. This definition is due to Benczur and Karger \cite{BK96}.

Spielman and Teng \cite{ST08:sparse} introduced the stronger definition of {\em spectral sparsification}. A  weighted graph $H=(V,E_H,w_H)$ is a spectral sparsifier of $G=(V,E_G,w_G)$ with error $\epsilon$ if
\begin{equation} \label{eq.spectral}
    \forall x \in \R^V \ \ \ (1-\epsilon) \ x^T L_G x \leq x^T L_H x \leq (1+\epsilon) \ x^T L_G x
\end{equation}
where $L_G$ is the Laplacian matrix of the graph $G$. If $A_G$ is the adjacency matrix of $G$ and $D_G$ is the diagonal matrix of weighted degrees, then the Laplacian matrix is $L_G = D_G-A_G$ and it has the property that, for every vector $x\in \R^V$,
\[ x^T L_G x = \sum_{(u,v) \in E_G} w_{u,v} \cdot (x_u - x_v)^2 \]
The definition of spectral sparsifier is  stronger than the definition of cut sparsifier because, if $x = {\mathbf 1}_S$ is the 0/1 indicator vector of a set $S$, then we have $x^TL_Gx = \cut_G(S)$. So we see that the definition in \eqref{eq.cut} is equivalent to a specialization of the definition of \eqref{eq.spectral} to the case of Boolean vectors $x \in \{ 0,1 \}^V$.

In all the known constructions of sparsifiers, the edge set $E_H$ of the sparsifier is a subset of the edge set $E_G$ of the graph $G$. We will take this condition to be part of the definition of sparsifier.

A cut sparsifier $H$ of a graph $G$ has, approximately, the same cut structure of $G$, so that, if we are interested in approximately solving a problem involving cuts or flows in $G$, we may instead solve the problem on $H$ and be guaranteed that an approximate solution computed for $H$ is also an approximate solution for $G$.

As the name suggests, for every graph $G$ it is possible to find a cut sparsifier  $H$ of $G$ which is very sparse, and running an algorithm on a sparse graph yields a faster running time than running it on $G$, if $G$ is not sparse itself.

A spectral sparsifier $H$ of $G$ has all the properties of a cut sparsifier, and, furthermore, it can be substituted for $G$ and it can accelerate computations on $G$ in some additional applications. For example, if we wish to solve a Laplacian linear system $L_Gx = b$, and $H$ is a good spectral sparsifier of $G$, then we can use $L_H$ as a preconditioner and solve $L_H^{-1}L_Gx = L_H^{-1}b$ instead. The condition number of $L_H^{-1}L_G$ will be small, making convergence fast, in return for solving the sparse problem $L_Hy = a$ once per iteration.

Benczur and Karger \cite{BK96} showed that, for every graph $G$,  a cut sparsifier with error $\epsilon$ having $O(\epsilon^{-2} n \log n)$ edges can be computed in nearly linear time. Spielman and Teng \cite{ST08:sparse} proved that a spectral sparsifier with error $\epsilon$ having $O(\epsilon^{-2} n (\log n)^{O(1)} )$ edges can be computed in nearly linear time. Spielman and Srivastava \cite{SS11} improved the number of edges that suffice to construct a spectral
sparsifier to $O(\epsilon^{-2} n\log n)$, and Batson, Spielman and Srivastava \cite{BSS09} reduced it to $O(
\epsilon^{-2} n)$.   Up to the constant in the big-Oh notation, the $O(\epsilon^{-2} n)$ bound is best possible, because every $\epsilon$ cut sparsifier of the clique (and therefore, since it is a stronger condition, every $\epsilon$ spectral sparsifier of the clique) requires
$\Omega(\epsilon^{-2} n)$ edges~\cite{ACKQWZ16}. While the construction of Batson, Spielman and Srivastava does not run in nearly linear time, there have been subsequent faster constructions with $O(\epsilon^{-2} n)$ edges running in nearly quadratic time \cite{ALO15} and nearly linear time \cite{LS17}.

In this paper we focus on the combinatorial problem of understanding the minimum number of edges that suffice to achieve cut and spectral sparsification, regardless of the efficiency of the construction. In particular, we aim to understand the best possible constant in the $\Theta(\epsilon^{-2} n)$ bound mentioned above.

Currently, the  construction (or even non-constructive existence proof) of cut sparsifiers for general graphs with the smallest number of edges is that due to Batson, Spielman and Srivastava, which also achieves spectral sparsification with the same parameters. In particular, prior to this work, there was no evidence that cut sparsification is ``easier'' than spectral sparsification, in the sense of requiring a smaller number of edges. In this paper we show that random $\log n$-regular graphs, with high probability, can be cut-sparsified with better parameters than they can be spectrally-sparsified, if one requires the sparsifier to use a subset of the edges of the graph to be sparisified. Under a conjecture of Srivastava and Trevisan, the same separation would apply to sparsifiers of the clique.

In the following, instead of referring to the number of edges in the sparsifier as a function of the error parameter $\epsilon$ and of the number of vertices $n$, it will be cleaner to refer to the error parameter $\epsilon$ as a function of the average degree $d$ of the sparsifier (that is, we call $dn/2$ the number of edges of the sparsifier).

The construction of Batson, Spielman and Srivastava achieves error $(2\sqrt 2)/\sqrt d$ with a sparsifier of average degree $d$, for general graphs. Batson, Spielman and Srivastava also show that every sparsifier of the clique of average degree $d$ has error at least $1/\sqrt d$.  Srivastava and Trevisan \cite{ST18} prove that every sparsifier of the clique of average degree $d$ and girth $\omega_n (1)$ (that is, with girth that grows with the number of vertices) that spectrally sparsifies the clique has error at least $(2-o_{n,d}(1))/\sqrt d$. Here, $o_{n,d}(1)$ denotes an error term such that there exists an expression depending only on $d$ after taking an $n \rightarrow \infty$ limit on the error term; a subsequent $d \rightarrow \infty$ limit then sends this expression to 0. For example, $1/\sqrt{d} + d/n$ is an expression that is $o_{n,d}(1)$. Furthermore, an appropriately scaled $d$-regular Ramanujan graph is a spectral sparsifier of the clique with error $(2+o_{n,d}(1))/\sqrt d$, so we will refer to $2/\sqrt d$ as the {\em Ramanujan bound} for sparsification. Srivastava and Trevisan conjecture that the Ramanujan bound is best possible for all graphs that sparsify the clique.

\begin{conjecture}[Srivastava and Trevisan] \label{conj.st} Every family of weighted graphs of average degree $d$
that are $\epsilon$ spectral sparsifiers of the clique satisfy $\epsilon > (2-o_d(1))/\sqrt d$.
\end{conjecture}

\subsection{Our Results}

Our first result is that it is possible to do better than the Ramanujan bound for cut sparsification of the clique.

In the following, we use $\gReg{n}{d}$ to denote the distribution over random $d$-regular multigraphs on $n$ vertices created by taking the disjoint union of $d$ random perfect matchings. We will always assume that $n$ is even.

\begin{restatable}[Main]{theorem}{main}
\label{th.main}
With $1-o_n(1)$ probability, a random regular graph drawn from $\gReg{n}{d}$, in which all edges are weighted $(n-1)/d$, is a $\left(2\sqrt{\frac 2\pi} + o_{n,d} (1)\right)/\sqrt d$ cut sparsifier of the clique, where $2\sqrt{\frac 2\pi}= 1.595\dots$
\end{restatable}

Together with Conjecture \ref{conj.st}, the above theorem (proved in Section~\ref{sec:separation}) gives a conditional separation between the error-density tradeoffs of cut sparsification versus spectral sparsification of the clique.

In order to achieve an unconditional separation, we prove generalizations of the result of Srivastava and Trevisan. Our first lower bound, which we think is of independent interest, applies to symmetric double stochastic matrices.

\begin{restatable}[Alon-Boppana for Symmetric Doubly Stochastic Matrices]{theorem}{doublystoc}
\label{th.doublystoc}
If $M$ is a symmetric $n\times n$ doubly stochastic matrix with $dn$ non-zero entries, then $M$ has a non-trivial eigenvalue of magnitude at least $\big( 2 - o_{n,d}(1) \big)/\sqrt d$.
\end{restatable}

The error term $o_{n,d}(1)$ is of the form $O\big( (\ln d)/d^{1/4} \big) + O \Big( d^{d^{1/4}}/n \Big) + O \big( 1/\sqrt{n} \big)$.
A doubly stochastic matrix is a matrix with non-negative entries such that every row and every column sums to one. A symmetric doubly stochastic matrix always has a ``trivial'' eigenvalue equal to one, corresponding to the eigenvector $(1,\ldots,1)$. The above theorem states the existence of at least one other eigenvalue whose absolute value is at least $2/\sqrt d - o_{n,d}(1)$. The Theorem can be restated as providing a spectral sparsification lower bound for weighted regular graphs, those such that all vertices have the same weighted degree.

We are also able to prove a $(2-o_{n,d}(1))/\sqrt d$ lower bound on the spectral sparsification error for certain classes graphs that have irregular weighted degree.

\begin{restatable}[Spectral Sparsification Lower Bound]{theorem}{speclower}
\label{th.speclower}
If $H$ is a graph with $n$ vertices and $dn/2$ edges, and such that at most $\sqrt n$ vertices of $H$ participates in odd cycles of length $\leq d^{1/4}$, and if $H$ is an $\epsilon$-spectral sparsifier of the clique, then $\epsilon > (2-o_{n,d}(1))/\sqrt d$
\end{restatable}

The error term is of the form $O\big( (\ln d)/d^{1/4} \big) + \exp \big( O(d^{1/4}) \big) /\sqrt n$. Our result has a more general form in which at most $B$ vertices participate in odd cycles of length up to $g$, and the error term is of the form $O \big( g/\sqrt d + 1/g + \exp(O(g)) \big)\cdot B/n$.
The underlying characteristic of this class of graphs that make them compatible with our techniques is that the odd powers of their adjacency matrices have trace bounded by $1/n^{\gamma}$ times the succeeding even power traces, for some $\gamma$.

In comparison, Srivastava and Trevisan prove such a lower bound for graphs of large girth. However, the argument is problematic to adapt even to graphs with one small cycle. Here we only need to rule out small odd cycles, and even the presence of some odd small cycles can be tolerated, provided that they do not involve too many vertices.

Using the fact that a random $\log n$-regular graph is, with high probability, an $O(1/\sqrt{\log n})$ spectral sparsifier of the clique,  that a random $\log n$-regular graph contains a random $d$-regular graph as a subgraph, and that, with high probability, all subgraphs of random $\log n$-regular graph satisfy the ``large odd pseudo-girth'' assumption of Theorem \ref{th.speclower}, we have our separation result as follows.

\begin{restatable}{theorem}{separation}
\label{th.separation}
Let $G$ be a random regular graph drawn from $\gReg{n}{\log n}$. Then with probability $1-o_n(1)$ over the choice of $G$ the following happens for every constant $d$:
\begin{enumerate}
    \item There is a weighted subgraph $H$ of $G$ with $dn/2$ edges such that $H$ is an $\epsilon$ cut sparsifier of $G$ with $\epsilon \leq (1.595\ldots + o_{n,d}(1))/\sqrt d$;
    \item For every weighted subgraph $H$ of $G$ with $dn/2$ edges, if $H$ is an $\epsilon$ spectral sparsifier of $G$ then $\epsilon \geq  (2- o_{n,d}(1))/\sqrt d$.
\end{enumerate}
\end{restatable}
The proof of this is found in Section~\ref{sec:separation}.

\subsection{Techniques in Cut Sparsification Result}

Our main result, Theorem \ref{th.main}, is established by analyzing cuts of linear size using rigorous techniques that have been derived from statistical physics \cite{jagannath2017unbalanced} and by analyzing sublinear size cuts using martingale concentration bounds.

For a fixed set $S$ of $k = \alpha n \leq n/2$ vertices, the average number of edges that leave $S$ in a random $d$-regular graph is $\frac d{n-1} \cdot k \cdot (n-k)$ and we are interested in showing that for every such set the deviation from the expectation is at most $\epsilon \frac d{n-1} \cdot k \cdot (n-k)$, for $\epsilon \le 1.595\dots/\sqrt d$.

\subsubsection{Bound for small sets}

One approach is to set up a martingale for each partition of vertices into two sets and then apply an Azuma-like inequality on the the number of edges cut. In this approach, it is better to study the deviation from the expectation of the number of  edges that are entirely contained in $S$. This is because, in a regular graph, the deviation from the expectation of the number of edges crossing the cut $(S,V-S)$ is entirely determined by the deviation from the expectation of the number of edges entirely contained in $S$, and the latter can be written as a sum of fewer random variables (that is, ${k \choose 2}$ versus $k\cdot (n-k)$), especially for small $k$. After setting up the appropriate Doob martingale, we can prove that the probability that the cut $(S,V-S)$ deviates from the expectation by more than $1.595\dots/\sqrt d$ times the expectation is at most $e^{-\Omega(n)}$ if $k\geq \Omega(n/\sqrt d)$ and at most $e^{-\Omega ( d k \log (n/dk) )}$ for $k \leq O(n / \sqrt d)$. In particular, there is an $\alpha_0 > 0$ such that for all $k\leq \alpha_0 n$ the probability of having a large deviation is much smaller than $1/{n\choose k}$, in a way that enables a union bound. These calculations are carried out in Section \ref{sec.sublinear}.

Unfortunately, such ``first moment'' calculations cannot be pushed all the way to $\alpha_0 = 1/2$. This is because our calculations with deviation bounds and union bounds are equivalent to estimating the average number of cuts that have a relative error (the ratio of the deviation from expectation to the expectation of the number of edges cut) bigger than $1.595\dots/\sqrt d$, with the goal of showing that such average number is much smaller than one. Unfortunately, the average number of balanced cuts that have a relative error bigger than $2/\sqrt d$ is bigger than one, so we cannot hope to get a separation from the spectral bounds with first moment calculations. Instead, we'll see that this distribution is extremely heavy-tailed: although the average number of such cuts is larger than 1, with very high probability there are 0 such cuts.

\subsubsection{Bound for large sets}
\label{sec:techniques-large-sets}

We then turn to techniques derived from statistical physics in order to analyze large cuts. To illustrate this approach, consider the classical problem of bounding the typical value of the max cut optimum in   Erd\H{o}s-R\'{e}nyi random graphs $\gErd{n}{1/2}$, up to $o(n^{1.5})$ error terms. This is equivalent to the problem of understanding the typical value of
\begin{equation}\label{eq.rad}
    \max_{\sigma \in \{ \pm 1 \}^n} \sigma^T M \sigma
\end{equation}
where $M$ is a random symmetric matrix with independent uniform $\pm 1$ entries off the diagonal and zero diagonal.

A first step is to prove, by an interpolation argument, that, up to lower order $o(n^{1.5})$ additive error, the optimum of \eqref{eq.rad} is the same as the optimum of
\begin{equation} \label{eq.sk}
  \max_{\sigma \in \{ \pm 1 \}^n} \sigma^T W \sigma
\end{equation}
where $W$ is a Wigner matrix, a random symmetric matrix with zero diagonal and independent and standard normally distributed off-diagonal entries.

Finding the optimum of \eqref{eq.sk} up to an additive error $o(n^{1.5})$ is a standard problem in statistical physics: it is the problem of determining the zero-temperature free energy of a spin-glass model called the  Sherrington-Kirkpatrick model, or SK model for short.

Parisi \cite{Parisi_1980} defined a family of differential equations, and presented a heuristic argument according to which the infimum of the solutions of those differential equations, would give the free energy of the SK model. That infimum  is now called the Parisi formula. Parisi's approach was extremely influential and widely generalized. Guerra \cite{guerra2003broken} rigorously proved that a solution to each of the differential equations gives an upper bound on the free energy, and, in a monumental work, Talagrand \cite{talagrand2006parisi} rigorously proved the stronger claim that the Parisi formula is equal to the free energy of the SK model. Talagrand's work was further generalized by Panchenko \cite{Pan14}.

Dembo, Montanari and Sen \cite{dembo2017extremal} proved an interpolation result showing that the solution to \eqref{eq.sk} can also be used to bound the max cut in random sparse graphs of constant average degree $d$, including both random $d$-regular graphs $\gReg{n}{d}$ and Erd\H{o}s-R\'{e}nyi random graphs $\gErd{n}{d/n}$. Jagannath and Sen \cite{jagannath2017unbalanced} proved interpolation theorems for the problem of determining the max cut out of sets of size $\alpha n$, for fixed constant $\alpha$, in $\gErd{n}{d/n}$ and in $\gReg{n}{d}$ graph, and they proved that the two models have different asymptotic bounds when $0 < \alpha < 1/2$.

In particular, to find the maximum (and the minimum) over all sets $S$ of cardinality $\alpha n$ of $\cut_G(S)$ in a random $d$-regular graph, Jagannath and Sen prove that one has to study
\begin{equation} \label{eq.js}
    \max_{\sigma \in S_n(\alpha)} \sigma^T \Pi^T W \Pi \sigma
\end{equation}
where $S_n(\alpha)$ is the subset of vectors $\sigma \in \{ \pm 1\}^n$ that contain exactly $\alpha n$ ones, and $\Pi = I - \frac 1n J$ is the matrix that projects on the space orthogonal to $(1,1,\ldots,1)$. The restriction to $S_n(\alpha)$ models the restriction to cuts $(S,V-S)$ where $|S| = \alpha n$, and the projection defines a matrix $\Pi^T W \Pi$ such that all rows and all columns sum to zero, in analogy to the fact that, in a regular graph, all rows and all columns of the adjacency matrix have the same sum.

Jagganath and Sen also define a Parisi-type family of differential equations and they rigorously prove that a solution to any of those equations provides an upper bound to \eqref{eq.js}. Since their goal is to compare cuts in regular graphs to cuts in Erd\H{o}s-R\'{e}nyi graphs, rather than bounding cut sizes in random regular graphs, they do not provide solutions to their Parisi-type equations. In Section \ref{sec.linear} we compute the {\em replica-symmetric} solution and get an explicit bound.

From the bound, we get that, for every fixed $\alpha$, with high probability,
sets of size $\alpha n$ in a random $d$-regular graph satisfy the definition of
$\epsilon$ cut sparsification of the clique with
\[ \epsilon \leq \left( 2 \ \sqrt{\frac {2} {\pi}}+ o_{n,d}(1) \right) \cdot \frac{1}{\sqrt{d}} = \frac{1.595\ldots + o_{n,d} (1) }{\sqrt d}  \]

A tight upper bound on $\epsilon$, which would come from an exact solution of \eqref{eq.js}, is likely to be $2/\sqrt{d}$ times the value of the Parisi formula evaluated at zero temperature and no external field (approximately $1.5264/\sqrt{d}$ \cite{PhysRevE.65.046137}), although we have not attempted to prove this.
This is motivated by Jagganath and Sen's generalization reducing to the original Parisi formula at zero temperature and no external field, parameters which correspond to maximum bipartitions.

\subsection{Techniques in Lower Bounds for Spectral Sparsification}

As discussed above, we established that a random $d$-regular graph is an $\epsilon$ cut sparsifier with $\epsilon \le (1.595\dots + o(1))/\sqrt d$. Under  Conjecture \ref{conj.st}, this gives a conditional separation between the error-vs-density tradeoff for cut sparsification of the clique compared to spectral sparsification of the clique.

Although we are not able to prove Conjecture \ref{conj.st}, we are able to make some new progress toward it.

The Alon-Boppana theorem states that if $A_H$ is the adjacency matrix of a $d$-regular graph $H$ on $n$ vertices, then $A_H$ has a non-trivial eigenvalue of magnitude at least $2\sqrt {d-1} - o_n(1)$. (The statement refers to a ``non-trivial'' eigenvalue to distinguish it from the ``trivial'' eigenvalue of value $d$ which is always present in the adjacency matrix of a $d$-regular graph.) If $H$ is a $d$-regular graph in which we weigh every edge by $(n-1)/d$, so that nodes have the same weighted degree as a clique, then the Alon-Boppana theorem tells us that, if we regard $H$ as a spectral sparsifier of the clique then the sparsification error is at least $2\sqrt{d-1}/d - o_n (1) = 2/\sqrt{d} - o_{n,d}(1)$.

This means that the Alon-Boppana theorem provides a sparsification lower bound for sparsifying the clique with graphs that are regular and in which all edges have the same weight. Equivalently, clique sparsification lower bounds can be seen as generalizations of the Alon-Boppana theorem to graphs that are not regular and whose edges are weighted in an arbitrary way.

The Alon-Boppana theorem has two known proofs, both described in the survey \cite{HLW06}. The original proof of Alon and Boppana constructs an explicit test vector orthogonal to $(1,\ldots,1)$ and proceeds by evaluating the quadratic form of such test vector. This proof is extended to the general sparsification setting by Srivastava and Trevisan \cite{ST18}, but their approach requires the graph to have large girth, and fails even if the graph has few small cycles. The other proof of the Alon-Boppana theorem, due to Friedman, proceeds by bounding the trace of a high power of the adjacency matrix of the graph.

This is the proof that we adapt to weighted graphs in this paper, and that allows us to prove Theorem \ref{th.doublystoc}. Our trace bound extends to the adjacency matrices of graphs that are almost regular from the point of view of weighted degrees (which can be assumed without loss of generality for sparsifiers of the clique).

In order to bound the sparsification error, however, it is not enough to find a non-trivial eigenvalue of the adjacency matrix, but we need to find a non-trivial eigenvalue of the difference between the Laplacian matrix of the graph and the Laplacian matrix of the clique. Our first step is to go from a trace bound on the adjacency matrix to an explicit test vector (actually, a test density matrix) of the adjacency matrix, and then evaluate the quadratic form on the difference of the Laplacians. The result is the desired lower bound provided that we can bound the trace of $A_H^{2\ell -1}$, for all $\ell$ up to $d^{1/4}$. This term is zero if $H$ has large odd girth (a relaxation of the large girth condition needed for the proof of Srivastava and Trevisan), and it is small enough for our purposes if $o(n)$ vertices participate in short odd cycles. This is how we prove Theorem \ref{th.speclower}.

The latter ``odd pseudo-girth'' condition is satisfied by several families of random regular graphs and Erd\H{o}s-R\'{e}nyi random graphs. In particular, random $\Delta_n$-regular graphs, for any choice of the degree $\Delta_n$ is of the order of $\log n$. For every fixed $d$, a random $\log n$-regular graph $G$ contains a random $d$-regular graph $H$, and we
also have that $G$ is a $O(1/\sqrt {\log n})$ spectral (and cut) sparsifier of the clique.
We can conclude that, with high probability, $G$ contains a weighted subgraph $H$ (a scaled random $d$-regular subgraph) with $dn/2$ edges that is a $1.595\ldots$ cut sparsifier of the clique, and hence of $G$ (up to negligible difference), but for every weighted subgraph $H$ of $G$ that is an $\epsilon$ spectral sparsifier of the clique (and hence, up to negligible difference) of $G$, we have that $\epsilon > \big( 2-o_{n,d}(1) \big) / \sqrt d$. This established the separation result of Theorem \ref{th.separation}.

\subsection{Additional Remarks and Open Problems}

The notions of cut sparsifier and of spectral sparsifier of the clique are interesting generalizations of the notion of expander graph, allowing for graphs that are possibly weighted and irregular. As with expander graphs, it seems worthwhile to study sparsifiers as fundamental combinatorial objects, beyond their applications to the design of efficient graph algorithms.

A proof of Conjecture \ref{conj.st} would give us a significant generalization of the Alon-Boppana theorem, and it would be a very interesting result.

It is plausible that the clique is the hardest graph to sparsify, both for cut sparsification and for spectral sparsification. This would mean that the error in the construction of Batson, Spielman and Srivastava can be improved from $2\sqrt 2/ \sqrt d$ to $2/\sqrt d$, up to lower order terms, and that there is a construction (or perhaps a non-constructive existence proof) of cut sparsifiers of general graphs with error smaller than $1.6/\sqrt d$, up to lower error terms. At present, unfortunately, there is no promising approach to construct (or non-constructively prove the existence) of cut sparsifiers of general graphs with error below $2/\sqrt d$, or even below $2\sqrt 2/ \sqrt d$. The techniques presented in this paper are not immediately generalizable to broader families of graphs as they are tailored to exploit symmetries of the clique. Achieving the aforementioned objectives will certainly require new innovations.

\section{Linear-sized cuts}
\label{sec.linear}

We show that random regular graphs are good cut sparsifiers of the clique over cuts with vertex set $S$ of linear size, so that $|S| = \alpha n$ for constant $\alpha$.

\begin{restatable}[Linear Set Regime for Cut Sparsification]{theorem}{largecuts}
\label{thm:large-cuts}
For every $d$ and $\beta \in (0, 1/2)$, with probability $1 - o_n(1)$ over random regular multigraphs $H \sim \gReg{n}{d}$, it is true for every subset $S$ of vertices satisfying $|S| = \alpha n$ with $\alpha \in [\beta, 1/2]$ that
\[\left|\frac{\cut_H(S)}{d\alpha(1-\alpha)n} - 1 \right| \le \frac{1}{\sqrt{d}}\left(2\sqrt{\frac{2}{\pi}} + o_{n,d}(1)\right),\]
where $\cS_{\alpha} = \{S \subseteq V \mid |S| = \alpha n\}$ and $2\sqrt{\frac{2}{\pi}}=1.595\ldots$.
\end{restatable}


First we refer to a lemma showing that the maximum cut with relative cut volume $\alpha$ concentrates around its expectation, so that we reduce the problem to understanding the expected value of the maximum cut.
We also state its version for minimum cuts, derived by negating and using sign symmetries in the statement and proof of the lemma, in accordance with \cite[Remark 1]{jagannath2017unbalanced}.
\begin{lemma}[Lemma 2.1 of \cite{jagannath2017unbalanced}]
\label{lem:js17}
For every $d$ and $\alpha \in (0,1)$,
\[ \Pr_{H \sim \gReg{n}{d}}\left[\left| \max_{S \in \cS_{\alpha}}\frac{1}{n}\cut_H(S) - \E_{H' \sim \gReg{n}{d}}\left[\max_{S' \in \cS_{\alpha}} \frac{1}{n}\cut_{H'}(S')\right]\right| > \eps \right] \le 2e^{-n\eps^2/d},  \]
\[ \Pr_{H \sim \gReg{n}{d}}\left[\left| \min_{S \in \cS_{\alpha}}\frac{1}{n}\cut_H(S) - \E_{H' \sim \gReg{n}{d}}\left[\min_{S' \in \cS_{\alpha}} \frac{1}{n}\cut_{H'}(S')\right]\right| > \eps\right] \le 2e^{-n\eps^2/d}.  \]
\end{lemma}

As discussed in Section~\ref{sec:techniques-large-sets}, we now invoke techniques of statistical mechanics developed in the study of spin glasses, specifically the SK model and its generalizations.

After the Parisi formula was proven to solve the SK model, Dembo, Montanari, and Sen \cite{dembo2017extremal} used interpolation techniques to show that the free energy of the SK model corresponds to the maximum or minimum bisection (equivalently, balanced cut) on random sparse graphs.
Sen \cite{sen2018optimization} generalized that interpolation to a family of combinational problems, including unbalanced cuts $\cut(S)$ where $|S|$ is a constant times $n$, as we study here, relating these problems to a generalization of the SK model.

The SK model has internal energy $\sigma^TW\sigma/\sqrt{n}$
for $W \in \R^{n \times n}$ a symmetric Wigner matrix with standard Gaussian entries on the off-diagonals\footnote{This definition corresponds to that used in \cite{jagannath2017unbalanced}, and is larger by a factor of $2$ than a convention used in some other places.} and zero on the diagonals, to be optimized over configurations $\sigma \in \{\pm 1\}^n$.
The generalization studies the optimization problem with the same matrix $W$ and the same configuration space $\{\pm 1\}^n$ but with internal energy
\[H_W^{(1)}(\sigma) = \frac{1}{\sqrt{n}}\sigma^T\Pi W\Pi\sigma,\]
where $\Pi$ is the orthogonal projection away from the all-ones vector.
In this model, finding the extremal cuts of a given relative vertex density $\alpha$ corresponds to optimizing that energy over the restricted set of configurations
\[S_{n}(\alpha) = \left\{\sigma \in \{\pm 1\}^n : \sum_{i} \sigma_i = n(2\alpha - 1)\right\}.\]
We may formulate this equivalently as optimizing
\[ H_W^{(0)}(\sigma) = \frac{1}{\sqrt{n}}\sigma^TW\sigma\]
over a different alphabet $\sigma \in \{\pm 1 - (2\alpha - 1)\}$, with graph cuts of relative vertex density $\alpha$ corresponding to the set of configurations
\[ A_n(T(\alpha), \eps_n) = \left\{\sigma \in \{\pm 1 - (2\alpha - 1)\} : \left|\sum_{i}\sigma_i^2 - T(\alpha)\right| < \eps_n\right\},\]
with $T(\alpha) = 4\alpha(1-\alpha)$ and setting $\eps_n = 0$ to achieve the equivalence.

Finally, Jagannath and Sen \cite{jagannath2017unbalanced} used an \emph{analytical annealing} approach to solve this generalized model, yielding the generalization of the Parisi formula stated here:
\begin{definition}
\label{def:ground-state-energy}
Let $\nu$ be a measure over $[0,T]$ of the form $\nu = m(t)\dd t + c\delta_T$ with $m(t)$ non-negative, non-decreasing, and everywhere right-continuous with left limits (cadlag), where $\dd t$ is the uniform measure and $\delta_T$ is the Dirac delta function at $t = T$.
Then for $\lambda \in \R$ and $T(\alpha) = 4\alpha(1-\alpha)$, we define the ground state energy functional
\[\cP^1_{T(\alpha)}(\nu, \lambda) = u_{\nu, \lambda}(0,0) - \lambda T(\alpha) - 2\int_{0}^{T(\alpha)}s\dd\nu(s) \]
where $u_{\nu, \lambda}$ is the solution to the differential equation with boundary condition
\[ \left\{\begin{array}{ll}
    \displaystyle \frac{\partial u}{\partial t} + 2\frac{\partial^2 u}{\partial x^2} + 2m(t) \left( \frac{\partial u}{\partial x} \right)^2 = 0, & (t,x) \in [0,T(\alpha)) \times \R, \\
    \displaystyle \vphantom{\sum^n} u(x,T(\alpha)) = \max_{\zeta \in \{\pm 1 - M\}} \zeta x + (\lambda + 2c)\zeta^2, &
\end{array}\right.  \]
where $M = 2\alpha - 1$.
\end{definition}
This definition reduces to the original Parisi formula at zero temperature and external field in the case that $T = 1$ and $M = 0$ and when the infimum over $\nu$ is taken. Furthermore, this generalized Parisi formula relates to average extremal cuts on random regular graphs in the following way.
\begin{theorem}[Combination of Theorem 1.2 and Lemma 2.2 of \cite{jagannath2017unbalanced}]
\label{thm:psi-to-random-graphs}
Let $T(\alpha) = 4\alpha(1-\alpha)$.
For all $\alpha$,
\[ \E_{H \sim \gReg{n}{d}} \,\max_{S \in \cS_{\alpha}}\, \left|\frac{1}{n}\cut_H(S) - d\alpha(1-\alpha)\right| \le \frac{\sqrt{d}}{4}\,\inf_{\nu,\lambda}\, \cP^1_{T(\alpha)}(\nu, \lambda) + o_d(\sqrt{d}).  \]
\end{theorem}
\begin{proof}
By \cite[Lemma 2.2]{jagannath2017unbalanced}, in the limit as $n \to \infty$,
\[ \E_{H'}\left[\max_{S' \in \cS_{\alpha}} \frac{1}{n}\cut_{H'}(S')\right] = d\alpha(1-\alpha) + \frac{1}{4}\sqrt{d}\E_W\left[\max_{\sigma \in S_n(\alpha)} \frac{1}{n}H_W^{(1)}(\sigma)\right] + o_d(\sqrt{d}).  \]
As alluded to in \cite[Remark 1]{jagannath2017unbalanced}, Lemma 2.2 of \cite{jagannath2017unbalanced} holds also for minimum cuts: this requires only changing some signs and invoking a few instances of sign-flip symmetry in the proof.
\[ \E_{H'}\left[\min_{S' \in \cS_{\alpha}} \frac{1}{n}\cut_{H'}(S')\right] = d\alpha(1-\alpha) - \frac{1}{4}\sqrt{d}\E_W\left[\max_{\sigma \in S_n(\alpha)} \frac{1}{n}H_W^{(1)}(\sigma)\right] - o_d(\sqrt{d}).  \]

By the equivalence described earlier in this section and the fact that $A_n(T(\alpha),0) \subseteq A_n(T(\alpha), \eps_n)$ for any sequence of $\eps_n > 0$,
\[\max_{\sigma \in S_n(\alpha)} \frac{1}{n}H_W^{(1)}(\sigma) \le \max_{\sigma \in A_n(T(\alpha), \eps_n)}\, \frac{1}{n}H_W^{(0)}(\sigma) \]

By \cite[Theorem 1.2]{jagannath2017unbalanced}, for $\eps_n \to 0$ slowly enough as $n \to \infty$, it holds that for all $T$,
\[ \liminf_{n\to \infty} \E_W\;\max_{\sigma \in A_n(T(\alpha), \eps_n)}\, \frac{1}{n}H_W^{(0)}(\sigma) = \inf_{\nu,\lambda} \cP^1_{T(\alpha)}(\nu, \lambda). \]
Although the statement of \cite[Theorem 1.2]{jagannath2017unbalanced} is stated asymptotically almost surely over random choice of $W$, it also holds in expectation over $W$.
\comment{because of the fact that the Legendre transform is monotone with respect to convex combination (that is, if $a+b = 1$ with $a \ge 0$ and $b \ge 0$, then $af^* + bg^* \le (af + bg)^*$, where $^*$ is the Legendre transform) and therefore expectation.
This combines with the fact that $\sup_T F_N(\beta, \xi; A_N)$ in the proof of \cite[Theorem 1.2]{jagannath2017unbalanced} is the Legendre transform with dual variables $T$ and $\lambda$, and the fact that Legendre transforms are always convex and therefore the Legendre transform of the expectation is monotone with respect to taking a limit.}
See Lemma~\ref{lem:js17-fix} in the appendix for details.
Combining the above equations subsequently yields the theorem statement.
\end{proof}

It is not yet known how to efficiently compute the exact value of the Parisi formula or its generalization.
We circumvent this issue by providing an upper bound, by choosing a particularly simple measure $\nu$ to bound the infimum $\inf_{\nu,\lambda} \cP^1_{T(\alpha)}(\nu, \lambda)$.
Specifically, the choice of $\nu = c\delta_T$ with $m(t) = 0$ is known as the \emph{replica-symmetric ansatz} \cite[Chapter 2]{malatesta2019random}, corresponding to the first of Parisi's original sequence of estimates.
\begin{lemma}
\label{lem:replica-symmetric-bound}
\[ \inf_{\nu,\lambda} \cP^1_{T(\alpha)}(\nu, \lambda)
\le 8\sqrt{\alpha(1-\alpha)}\cdot \frac{1}{\sqrt{2\pi}}e^{-(\erf^{-1}(2\alpha -1))^2},\]
where $\erf$ is the Gauss error function $\erf(x) = \frac{1}{\sqrt{\pi}}\int_{-x}^x e^{-x^2}\dd x$.
\end{lemma}
\begin{proof}
First we express $\int_0^T s\dd\nu(s) = cT + \int_0^T tm(t)\dd t$ and reparameterize $\hat{\lambda} = \lambda + 2c$ so that we can write
\[ \inf_{\nu,\lambda} \cP^1_{T(\alpha)}(\nu, \lambda) = \inf_{\nu,\hat{\lambda}} \hat{u}_{\nu, \hat{\lambda}}(0,0) - \hat{\lambda}T - 2 \int_0^T tm(t)\dd t\]
where $\hat{u}_{\nu, \hat{\lambda}}$ is the solution to
\[ \left\{\begin{array}{ll}
    \displaystyle \frac{\partial u}{\partial t} + 2\frac{\partial^2 u}{\partial x^2} + 2m(t) \left( \frac{\partial u}{\partial x} \right)^2 = 0, & (t,x) \in [0,T) \times \R, \\
    \displaystyle \vphantom{\sum^n} u(x,T) = \max_{\zeta \in \{\pm 1 - M\}} \zeta x + \hat{\lambda}\zeta^2, &
\end{array}\right.  \]
with $M = 2\alpha - 1$.

By taking $\nu(t) = c\delta_T$ so that $m(t) = 0$, we can upper-bound the infimum over $\nu$, so that
\[ \inf_{\nu,\lambda} \cP^1_{T(\alpha)}(\nu, \lambda) \le \inf_{\hat{\lambda}} \hat{u}_{\hat{\lambda}}(0,0) - \hat{\lambda}T\]
and  $\hat{u}_{\hat{\lambda}}$ is the solution to
\[ \left\{\begin{array}{ll}
    \displaystyle \frac{\partial u}{\partial t} + 2\frac{\partial^2 u}{\partial x^2} = 0, & (t,x) \in [0,T) \times \R, \\
    \displaystyle \vphantom{\sum^n} u(x,T) = \max_{\zeta \in \{\pm 1 - M\}} \zeta x + \hat{\lambda}\zeta^2. &
\end{array}\right.  \]
By reparameterizing $t$ as $-t$ here, we can see that $u(x,0)$ is simply the result of evolving $u(x,T)$ according to the heat equation with diffusivity constant $2$ for a time of $T$.
Evolution of the heat equation with diffusivity $k$ over a time of $T$ is equivalent to convolution with the Gaussian heat kernel $\exp(-x^2/(4kT))/\sqrt{4\pi kT}$ \cite[Chapter 2.3]{evans10}, so
\[ \hat{u}_{\hat{\lambda}}(x,0) = \frac{1}{\sqrt{8\pi T}} \int_{-\infty}^{\infty} e^{-z^2/(8 T)}\left(\max_{\zeta \in \{\pm 1 - M\}} \zeta (z+x) + \hat{\lambda}\zeta^2\right)\dd z. \]
Thus
\[ \inf_{\nu,\lambda} \cP^1_{T(\alpha)}(\nu, \lambda)
\le \inf_{\hat{\lambda}} \frac{1}{\sqrt{8\pi T}} \int_{-\infty}^{\infty} e^{-z^2/(8 T)}\left(\max_{\zeta \in \{\pm 1 - M\}} \zeta z + \hat{\lambda}\zeta^2\right)\dd z - \hat{\lambda}T.\]
Now we calculate
\begin{align*}
    \max_{\zeta \in \{\pm 1 - M\}} \zeta z + \hat{\lambda}\zeta^2
    =& \max \left(-z - Mz + \hat{\lambda}(1 + 2M + M^2), z - Mz + \hat{\lambda}(1 - 2M + M^2)\right)
    \\=& -Mz +\hat{\lambda}(1 + M^2) + \max (-z + 2M\hat{\lambda}, z -2M\hat{\lambda})
    \\=& -Mz +\hat{\lambda}(1 + M^2) + |z -2M\hat{\lambda}|,
\end{align*}
so that
\[ \inf_{\nu,\lambda} \cP^1_{T(\alpha)}(\nu, \lambda)
\le \inf_{\hat{\lambda}} \frac{1}{\sqrt{8\pi T}} \int_{-\infty}^{\infty} e^{-z^2/(8T)}\left(-Mz +\hat{\lambda}(1 + M^2) + \left|z -2M\hat{\lambda}\right|\right)\dd z - \hat{\lambda}T.\]
Partially evaluating the integral using the facts that a Gaussian probability density function integrates to 1 and, by oddness of the integrand, $\int_{-\infty}^{\infty} ze^{-z^2/(8 T)} \dd z = 0$,
\[ \inf_{\nu,\lambda} \cP^1_{T(\alpha)}(\nu, \lambda)
\le \inf_{\hat{\lambda}} \frac{1}{\sqrt{8\pi T}} \int_{-\infty}^{\infty} e^{-z^2/(8 T)}\left|z -2M\hat{\lambda}\right|\dd z + \hat{\lambda}(1 - T + M^2).\]
Employing a change of variables $z \to 2\sqrt{T}z$ to write the integral in terms of the normal Gaussian probability density $\phi(z) = \frac{1}{\sqrt{2\pi}} e^{-z^2/2}$ and also applying the identity $1 - T = 1 + 4\alpha^2 - 4\alpha = M^2$,
\[ \inf_{\nu,\lambda} \cP^1_{T(\alpha)}(\nu, \lambda)
\le \inf_{\hat{\lambda}} \int_{-\infty}^{\infty} \phi(z)\left|2\sqrt{T}z -2M\hat{\lambda}\right|\dd z + 2\hat{\lambda}M^2.\]
Focusing now on the integral,
\begin{align*}
    &{}\int_{-\infty}^{\infty} \phi(z)\left|2\sqrt{T}z -2M\hat{\lambda}\right|\dd z
    \\={}& \int_{M\hat{\lambda}/\sqrt{T}}^{\infty} \phi(z)\left(2\sqrt{T}z -2M\hat{\lambda}\right)\dd z
    + \int_{-\infty}^{M\hat{\lambda}/\sqrt{T}} \phi(z)\left(-2\sqrt{T}z +2M\hat{\lambda}\right)\dd z
    \\={}& \int_{-\infty}^{-M\hat{\lambda}/\sqrt{T}} \phi(z)\left(-2\sqrt{T}z -2M\hat{\lambda}\right)\dd z
    + \int_{-\infty}^{M\hat{\lambda}/\sqrt{T}} \phi(z)\left(-2\sqrt{T}z +2M\hat{\lambda}\right)\dd z,
\end{align*}
where we negated and flipped the limits of the first integral, which is equivalent to negating the odd part of the integrand while preserving the even part. Continuing to integrate, letting $\Phi(z)$ denote the Gaussian cumulative density function,
\begin{align*}
    ={}& \int_{-\infty}^{-M\hat{\lambda}/\sqrt{T}} -2\sqrt{T}z\,\phi(z)\dd z
    + \int_{-\infty}^{M\hat{\lambda}/\sqrt{T}} -2\sqrt{T}z\,\phi(z)\dd z
    + 2M\hat{\lambda}\int_{-M\hat{\lambda}/\sqrt{T}}^{M\hat{\lambda}/\sqrt{T}} \phi(z)\dd z,
    \\={}& \left[2\sqrt{T}\,\phi(z)\right]_{-\infty}^{-M\hat{\lambda}/\sqrt{T}} + \left[2\sqrt{T}\,\phi(z)\right]_{-\infty}^{M\hat{\lambda}/\sqrt{T}}
    + 2M\hat{\lambda}\left(\Phi(M\hat{\lambda}/\sqrt{T}) - \Phi(-M\hat{\lambda}/\sqrt{T})\right)
    \\={}& 4\sqrt{T}\,\phi(M\hat{\lambda}/\sqrt{T}) + 2M\hat{\lambda}\erf(M\hat{\lambda}/\sqrt{2T}),
\end{align*}
where we used evenness of $\phi$ and the fact that $\Phi(x) - \Phi(-x) = \erf(x/\sqrt{2})$ in the last step.
So, putting this evaluation of the integral into our previous expression,
\[ \inf_{\nu,\lambda} \cP^1_{T(\alpha)}(\nu, \lambda)
\le \inf_{\hat{\lambda}} 4\sqrt{T}\,\phi(M\hat{\lambda}/\sqrt{T}) + 2M\hat{\lambda}\erf(M\hat{\lambda}/\sqrt{2T}) + 2\hat{\lambda}M^2.\]
By finding the critical point of this expression with respect to $\hat{\lambda}$, we find a value of
$\hat{\lambda} = -\sqrt{2T}\erf^{-1}(M)/M$.
Using this value for $\hat{\lambda}$,
\begin{align*}
\inf_{\nu,\lambda} \cP^1_{T(\alpha)}(\nu, \lambda)
&\le 4\sqrt{T}\,\phi(-\sqrt{2}\erf^{-1}(M)) - 2M\hat{\lambda}M + 2\hat{\lambda}M^2
\\&= 4\sqrt{T}\,\phi(\sqrt{2}\erf^{-1}(M)). \qedhere
\end{align*}
\end{proof}

We calculate the largest concrete value attained by the upper bound of the preceding lemma:
\begin{lemma}
\label{lem:ground-state-analysis}
For all $\alpha \in (0,1)$,
\[\frac{\inf_{\nu,\lambda} \cP^1_{T(\alpha)}(\nu, \lambda)}{4\alpha(1-\alpha)} \le 2\sqrt{\frac{2}{\pi}} = 1.595... \]
\end{lemma}
\begin{proof}
By Lemma~\ref{lem:replica-symmetric-bound}, for $\alpha \in (0,1)$,
\[ \frac{\inf_{\nu,\lambda} \cP^1_{T(\alpha)}(\nu, \lambda)}{4\alpha(1-\alpha)} \le \frac{2}{\sqrt{\alpha(1-\alpha)}}\cdot \frac{1}{\sqrt{2\pi}}e^{-(\erf^{-1}(2\alpha -1))^2}
:= f(\alpha).\]
Evaluated at $\alpha = 1/2$, this is equal to $2\sqrt{2/\pi}$, so we just need to show that the upper bound $f(\alpha)$ is maximized at $\alpha = 1/2$.

First we reparameterize $g(M) = f(\alpha)$ with $M = 2\alpha - 1$ so that
\[g(M) = \frac{1}{\sqrt{2\pi(1 - M^2)}}e^{-(\erf^{-1}(M))^2}\]
and we want to show that $g$ is maximized at $0$.
Using the product rule to take the derivative of $g$, since $\frac{d}{dM}\frac{1}{\sqrt{1-M^2}} = \frac{M}{(1-M^2)^{3/2}}$ and $\frac{d}{dM}e^{-(\erf^{-1}(M))^2} = -\sqrt{\pi}\erf^{-1}(M)$,
\[g'(M) = \frac{1}{\sqrt{2\pi(1-M^2)}}\left(\frac{Me^{-(\erf^{-1}(M))^2}}{1-M^2} - \sqrt{\pi}\erf^{-1}(M)\right).\]
We take another monotonic reparameterization, introducing $\erf(x)$ for $M$:
\[g'(\erf(x)) = \frac{1}{\sqrt{2\pi(1-\erf(x)^2)}}\left(\frac{\erf(x)e^{-x^2}}{1-\erf(x)^2} - \sqrt{\pi}x\right).\]
By Polya \cite[Equation 1.5]{polya1945remarks}, $\erf(x) < \sqrt{1-e^{-4x^2/\pi}}$ so that $1-\erf(x)^2 \ge e^{-4x^2/\pi}$, so that, for $x < 0$ when $\erf(x) < 0$,
\[g'(\erf(x)) \ge \frac{1}{\sqrt{2\pi(1-\erf(x)^2)}}\left(\erf(x)e^{(4/\pi-1)x^2} - \sqrt{\pi}x\right).\]
And by Neuman \cite[Corollary 4.2]{neuman2013inequalities}, $\erf(x) \ge \frac{2x}{\sqrt{\pi}}e^{-x^2/3}$, so when $x < 0$,
\[g'(\erf(x)) \ge \frac{1}{\sqrt{2\pi(1-\erf(x)^2)}}\left(\frac{2x}{\sqrt{\pi}}e^{(4/\pi-4/3)x^2} - \sqrt{\pi}x\right).\]
And as $e^{(4/\pi-4/3)x^2} \le 1$, this makes it clear that $g'(\erf(x))$ is positive when $x$ is negative, which means that $g$ is increasing on the negative part of its domain, which by evenness of $g$ means that $g$ is maximized at $0$.
\end{proof}

We have all the ingredients necessary to prove the main theorem stated at the beginning of this section.

\largecuts*

\begin{proof}
By Theorem~\ref{thm:psi-to-random-graphs}, for every $\alpha \in (0,1)$,
\[ \E_{H \sim \gReg{n}{d}} \,\max_{S \in \cS_{\alpha}}\, \frac{1}{n}\cut_H(S) - d\alpha(1-\alpha) \le \frac{\sqrt{d}}{4}\,\inf_{\nu,\lambda}\, \cP^1_{T(\alpha)}(\nu, \lambda) + o_d(\sqrt{d}),  \]
\[ \E_{H \sim \gReg{n}{d}} \,d\alpha(1-\alpha) - \min_{S \in \cS_{\alpha}}\, \frac{1}{n}\cut_H(S) \le \frac{\sqrt{d}}{4}\,\inf_{\nu,\lambda}\, \cP^1_{T(\alpha)}(\nu, \lambda) + o_d(\sqrt{d}).  \]
Combining the above with Lemma~\ref{lem:js17}, we see, for every $\alpha \in (0,1)$,
\[ \Pr_{H \sim \gReg{n}{d}}\left[ \max_{S \in \cS_{\alpha}} \left| \frac{1}{n}\cut_H(S) - d\alpha(1-\alpha)\right| > \eps + \frac{\sqrt{d}}{4}\,\inf_{\nu,\lambda}\, \cP^1_{T(\alpha)}(\nu, \lambda) + o_d(\sqrt{d}) \right] \le 2e^{-n\eps^2/d}.  \]
By choosing $\eps_n = 1/\sqrt{\log(n)}$ in Lemma~\ref{lem:js17} and then putting it through a union bound over all $\le n$ possible values of $\alpha$, with probability $1-e^{\Omega(n/\log n)}$ for all $\alpha \in [\beta,1/2]$,
\[ \max_{S \in \cS(\alpha)} \left| \frac{1}{n}\cut_H(S) - d\alpha(1-\alpha)\right| > \frac{\sqrt{d}}{4}\,\inf_{\nu,\lambda}\, \cP^1_{T(\alpha)}(\nu, \lambda) + o_{n,d}(\sqrt{d}).  \]
And by using Lemma~\ref{lem:ground-state-analysis} to substitute for the value of $\inf_{\nu,\lambda}\, \cP^1_{T(\alpha)}(\nu, \lambda)$, we see that
\[ \max_{S \in \cS(\alpha)} \left| \frac{1}{n}\cut_H(S) - d\alpha(1-\alpha)\right| > 2\sqrt{\frac{2}{\pi}}\,\alpha(1-\alpha)\sqrt{d} + o_{n,d}(\sqrt{d}). \qedhere \]
\end{proof}

\section{Analysis for small cuts}
\label{sec.sublinear}

In this section, we demonstrate that the number of edges crossing a cut $(S, V-S)$ deviates no more from its expectation than by a $\frac{1.5}{\sqrt{d}}$ factor with high probability when $\lvert S \rvert$ is small.

\begin{restatable}[Small Set Regime for Cut Sparsification]{theorem}{smallcuts}
\label{thm:small-cuts-crossing}
For all sufficiently large $n \geq 0$ and constant $d \geq 0$ such that, for any $S \subset V$ where $\lvert S \rvert = k$ and $k \leq \frac{n}{100}$, a sample $H \sim \gReg{n}{d}$ admits with probability at least $1 - 2\binom{n}{k}^{-1.01}$
\begin{equation*}
\bigg\lvert \frac{\cut_H(S)}{\E_{H \sim \gReg{n}{d}} [\cut_H(S)]} - 1 \bigg\rvert \leq \frac{1.5}{\sqrt{d}}
\end{equation*}
\end{restatable}

To be sure, the exact value of the constant 1.5 is unimportant; it suffices to demonstrate the inequality in Theorem~\ref{thm:small-cuts-crossing} using any constant smaller than $2\sqrt{\frac{2}{\pi}}$ in Theorem~\ref{thm:large-cuts}. Our analysis will require the use of a Doob martingale.

\begin{definition}
Given random variables $A$ and $(Z_\ell)_{\ell=1}^N$ sampled from a common probability space, their associated \emph{Doob martingale} is given by random variables $(X_\ell)_{\ell=0}^N$ where $X_0 = \E[A]$ and
\begin{equation*}
X_\ell = \E [ A \;\vert\; Z_1, \ldots, Z_\ell ]
\end{equation*}
\end{definition}

We note that $(Z_\ell)$ is often called the \emph{filtration} that $(X_\ell)$ is defined with respect to. For a Doob martingale $(X_\ell)_{\ell=0}^N$, we denote its \emph{martingale difference sequence} by $(Y_\ell)$ where $Y_\ell = X_\ell - X_{\ell-1}$ and its \emph{quadratic characteristic sequence} by $(\langle X \rangle_\ell)$ where
\begin{equation*}
\langle X \rangle_\ell = \sum_{r=1}^\ell \E \big[ Y_r^2 \;\vert\; Z_1 \ldots Z_{r-1} \big]
\end{equation*}

As mentioned previously, the small cuts analysis will quantify the number of edges contained entirely within a cut and use the fact that, in a regular graph, the number of edges across a cut is uniquely determined by the number of edges within the cut. For a graph $H$, we will denote $e_H(S)$ by the number of edges $e \in E_H$ with both endpoints contained within $S \subseteq V$. When $H$ is sampled from a distribution, it is understood that $e_H(S)$ is a random variable.


\subsection{Martingale Construction}

Consider $H$ a random regular graph drawn from $\gReg{n}{d}$. Enumerate its vertices by $i \in [n]$, and its constituent matchings by $m \in [d]$. For $S \subset V$ of size $\lvert S \rvert = k$, we will assume without loss of generality that $S = \{ 1, \ldots, k \}$. Next, consider the sequence of matching-vertex pairs $\big( (m_\ell, i_\ell) \big)_{\ell=1}^N$ enumerating each $(m, i) \in [d] \times [k - 1]$ where $N = d \cdot (k-1)$. Let us now define the sequence of random variables $(Z_\ell)_{\ell=1}^N$ where $Z_{\ell} = Z_{(m_\ell, i_\ell)} \in V$ is the vertex that matching $m_\ell$ matches $i_\ell \in V$ to in $H$. Note that
\begin{equation*}
e(S) = \sum_{\ell=1}^N \bone \{ Z_\ell \in [k] \textup{ and } Z_\ell > i_\ell \}
\end{equation*}

We construct the Doob martingale on $e(S)$ using $(Z_\ell)$ as a filtration. The \emph{matched edge-vertex reveal martingale} $(X_\ell)_{\ell=0}^N$ is given by $X_\ell = \E [e(S) \;\vert\; Z_1, \ldots, Z_\ell]$. One should think of this martingale as counting the number of edges contained within $S$. As an increasing number of $Z_\ell$ are conditioned on, information regarding what edges exist in $H$ is revealed in an \emph{ordered} way. The order in which an edge is revealed is given by the enumeration of the vertices adjacent to the edge, and the matching the edge belonged to when $H$ was first sampled from $d$ random matchings. Additionally, notice that vertex $k$ is excluded from such pairs $(m_\ell, i_\ell)$. This is because $m_\ell$ can only match $k$ to $i_\ell < k$ for the edge to be contained in $S$. Consequently, revealing edges adjacent to $\{ 1, \ldots, k-1 \}$ suffices to uniquely determine $e(S)$. 

Our analysis of $(X_\ell)$ will now proceed as follows. We first determine bounds on the martingale difference and quadratic characteristic of $(X_\ell)$. These bounds are then used by a standard martingale concentration result to argue that the number of edges contained within $S$ cannot deviate far from its expectation. Finally, we complete the proof of Theorem~\ref{thm:small-cuts-crossing} by using the fact that concentration in the number of edges within $S$ immediately implies concentration in the number of edges in $\cut_H(S)$ when $H$ is a random $d$ regular graph.


\subsection{Properties of the Martingale}

To bound the martingale difference and quadratic characteristic of $(X_\ell)$, we examine how $e(S)$ behaves as an increasing number of $Z_\ell$ are conditioned on. We say that $\{ z_1, \ldots, z_\ell \} \subseteq [n]$ is a \emph{valid realization} of $Z_\ell$ if there exists a $d$ regular graph $H$ such that each $(i_\ell, z_\ell) \in E_H$. When $z_1, \ldots, z_\ell$ are deterministically provided, we can define the following quantities.
\begin{enumerate}
\item $a_\ell = a_\ell(z_1, \ldots, z_\ell)$ is the number of remaining vertices in $S$ that remain unmatched as a function of $z_1, \ldots, z_\ell$. We denote $a_0 = \lvert S \rvert = k$.

\item $b_\ell = b_\ell(z_1, \ldots, z_\ell)$ is the number of remaining vertices in $V$ that remain unmatched as a function of $z_1, \ldots, z_\ell$. We denote $b_0 = \lvert V \rvert = n$.
\end{enumerate}

We will also consider $a_\ell(z_1, \ldots, z_{\ell-1}, Z_\ell)$ and $b_\ell(z_1, \ldots, z_{\ell-1}, Z_\ell)$ where $Z_\ell$ is sampled according to the filtration specified in $X_\ell$. In this case, $a_\ell$ and $b_\ell$ are random variables distributed according to that of the random variable $Z_\ell$. When $z_1, \ldots, z_\ell$ are a valid realization, we can demonstrate a bound on the ratio $\frac{a_\ell}{b_\ell}$.

\begin{lemma}\label{lem:helper}
Let $H \sim \gReg{n}{d}$ be a random regular graph, $S \subseteq V$ such that $\lvert S \rvert = k < \frac{n}{2}$, and $N = d \cdot (k-1)$. For any $0 \leq \ell \leq N$ and valid realization $z_1, \ldots, z_\ell$, it happens that
\begin{equation*}
\frac{a_\ell}{b_\ell} \leq \frac{k}{n}
\end{equation*}
\end{lemma}
\begin{proof}
We proceed via induction on $\ell$. For the base case, $\ell = 0$ implies we have $\frac{a_0}{b_0} = \frac{k}{n}$. Let us now assume the lemma holds for $\ell-1$. Notice that any choice of $z_\ell$ admits one of three cases.
\begin{enumerate}
\item $z_\ell \in [k]$ and $z_\ell > i_\ell$. This corresponds to $z_\ell$ revealing the existence of an edge not previously known to be in $S$ when considering only $z_1, \ldots, z_{\ell-1}$. Hence $a_\ell = a_{\ell-1} - 2$ and $b_\ell = b_{\ell-1} - 2$ and
\begin{equation*}
\frac{a_\ell}{b_\ell}
= \frac{a_{\ell-1} - 2}{b_{\ell-1} - 2}
\leq \frac{a_{\ell-1}}{b_{\ell-1}}
\leq \frac{k}{n}
\end{equation*}

with the last inequality following by the inductive hypothesis.

\item $z_\ell \in [k]$ however $z_\ell < i_\ell$. This corresponds to $i_\ell$ having already been matched to $j \in [k]$ as revealed by $z_j$ for $j < \ell$. Thus, $a_{\ell} = a_{\ell-1}$ and $b_{\ell} = b_{\ell-1}$ and the inductive hypothesis is maintained.

\item $z_\ell \notin [k]$ however $z_\ell > i_\ell$. This corresponds to $m_\ell$ matching $i_\ell$ to a vertex not in $S$. Thus $a_{\ell} = a_{\ell} - 1$ and $b_{\ell} = b_{\ell} - 2$ and so
\begin{equation*}
\frac{a_\ell}{b_\ell}
= \frac{a_{\ell-1} - 1}{b_{\ell-1} - 2}
\leq \frac{a_{\ell-1} - 1}{b_{\ell-1} - n/k}
= \frac{a_{\ell-1} - k/n \cdot n/k}{b_{\ell-1} - n/k}
< \frac{n}{k}
\end{equation*}

where the second inequality follows as $k \leq \frac{n}{2}$ and the last inequality holds by the following principle: $\frac{p}{q} < r$ implies $\frac{p - rw}{q - w} < r$ for all $p, q, r, w \in \mathbb{Z}_{\geq 0}$ and we choose $p = a_{\ell-1}$, $q = b_{\ell-1}$, $r = \frac{k}{n}$, and $w = \frac{n}{k}$.
\end{enumerate}

In all cases, we have that the lemma holds for $\ell$, thus completing the induction.
\end{proof}

We now bound the martingale difference of $(X_\ell)$.

\begin{lemma}\label{lem:martingale-diff}
Let $H \sim \gReg{n}{d}$ be a random regular graph, $S \subseteq V$ such that $\lvert S \rvert = k < \frac{n}{2}$, and $N = d \cdot (k - 1)$. Then $Y_\ell$ associated with $(X_\ell)_{\ell=0}^N$ admits $\lvert Y_\ell \rvert \leq 1$ for all $i \in [N]$.
\end{lemma}
\begin{proof}
As the $d$ constituent matchings of $H$ are sampled independently and uniformly at random, it suffices to assume $d = 1$, and hence $N = k - 1$. Now let $\phi(a, b)$ be the expected number of edges contained inside a subset of $a$ vertices in a uniformly sampled perfect matching on $b$ vertices. $\phi(a, b)$ is the quantity
\begin{equation*}
\phi(a, b) = \binom{a}{2} \cdot \frac{1}{b - 1}
\end{equation*}

For a given $\ell$, we begin by fixing a valid realization of random variables $Z_1 = z_1, \ldots, Z_\ell = z_\ell$ and observe that $X_{\ell-1}$ can be computed as
\begin{align*}
X_{\ell-1}
&= \E [ e(S) \;\vert\; Z_1 = z_1, \ldots, Z_{\ell-1} = z_{\ell-1}] \\
&= \E \bigg[ \sum_{r=1}^N \bone \{ Z_r \in [k] \textup{ and } Z_r > r \} \; \bigg\vert \; Z_1 = z_1, \ldots, Z_{\ell-1} = z_{\ell-1} \bigg] \\
&= \sum_{r=1}^{\ell-1} \bone \{ z_r \in [k] \textup{ and } z_r > r \} + \phi(a_{\ell-1}, b_{\ell-1})
\end{align*}

where we have used linearity of expectations to separate terms of $e(S)$ that have been conditioned to be $z_r$, and those that remain random. $X_\ell$ is similarly given by the following.
\begin{equation*}
X_{\ell}
= \sum_{r=1}^{\ell} \bone \{ z_r \in [k] \textup{ and } z_r > r \} + \phi(a_{\ell}, b_{\ell})
\end{equation*}

We can now compute $Y_\ell$ as
\begin{equation*}
Y_\ell
= X_\ell - X_{\ell-1}
= \bone \{ z_\ell \in [k] \textup{ and } z_\ell > \ell \} + \big( \phi(a_{\ell}, b_{\ell}) - \phi(a_{\ell-1}, b_{\ell-1}) \big)
\end{equation*}

Let us denote $w_\ell = \bone \{ z_\ell \in [k] \textup{ and } z_\ell > \ell \}$. It is either the case that $w_\ell = 1$ or $w_\ell = 0$. Assuming $w_\ell = 1$, we first demonstrate that $Y_\ell \leq 1$. In this case, vertex $\ell$ is adjacent to $z_\ell \in S$. Consequently, $a_\ell = a_{\ell-1} - 2$ and $b_\ell = b_{\ell-1} - 2$ and we have
\begin{align*}
Y_\ell
&= w_\ell + \big( \phi(a_{\ell-1} - 2, b_{\ell-1} - 2) - \phi(a_{\ell-1}, b_{\ell-1}) \big) \\
&= 1 + \binom{a_{\ell-1} - 2}{2} \cdot \frac{1}{b_{\ell-1} - 3} - \binom{a_{\ell-1}}{2} \cdot \frac{1}{b_{\ell-1} - 1} \\
&= 1 + \binom{a_{\ell-1} - 2}{2} \cdot \bigg( \frac{1}{b_{\ell-1} - 3} - \frac{1}{b_{\ell-1} - 1} \bigg) - \frac{2 a_{\ell-1} - 3}{b_{\ell-1} - 1} \\
&\leq 1 + \frac{2 a_{\ell-1} - 3}{b_{\ell-1} - 1} - \frac{2 a_{\ell-1} - 3}{b_{\ell-1} - 1} \\
&= 1
\end{align*}
as required. Completing the analysis for $w_\ell = 1$, we demonstrate that $Y_\ell \geq 0$.
\begin{align*}
Y_\ell
&= 1 + \binom{a_{\ell-1} - 2}{2} \cdot \bigg( \frac{1}{b_{\ell-1} - 3} - \frac{1}{b_{\ell-1} - 1} \bigg) - \frac{2 a_{\ell-1} - 3}{b_{\ell-1} - 1} \\
&\geq 1 - \frac{2 a_{\ell-1}}{b_{\ell-1}} \\
&\geq 1 - \frac{2k}{n}
\end{align*}

The last inequality follows from an application of Lemma~\ref{lem:helper}. Suppose now that $w_\ell = 0$. Since $(X_\ell)_{\ell=1}^N$ is a Doob martingale, $\E [Y_\ell] = 0$ for all $\ell$. This implies $Y_\ell < 0 < 1$ since, in fact, $Y_\ell > 0$ whenever $w_\ell = 1$. All that remains to demonstrate is that $Y_\ell > -1$. Observe that $w_\ell = 0$ implies one of two cases.
\begin{enumerate}
\item $z_\ell \in [k]$ however $z_\ell < \ell$. Then $a_\ell = a_{\ell-1}$ and $b_\ell = b_{\ell-1}$ implying $Y_\ell = 0$.

\item $z_\ell \notin [k]$ however $z_\ell > \ell$. Then $a_\ell = a_{\ell-1} - 1$ and $b_\ell = b_{\ell-1} - 2$. We then compute $Y_\ell$ as
\begin{align*}
Y_\ell
&= w_\ell + \big( \phi(a_{\ell-1} - 1, b_{\ell-1} - 2) - \phi(a_{\ell-1}, b_{\ell-1}) \big) \\
&= \binom{a_{\ell-1} - 1}{2} \cdot \frac{1}{b_{\ell-1} - 3} - \binom{a_{\ell-1}}{2} \cdot \frac{1}{b_{\ell-1} - 1} \\
&= \binom{a_{\ell-1} - 1}{2} \cdot \bigg( \frac{1}{b_{\ell-1} - 3} - \frac{1}{b_{\ell-1} - 1} \bigg) - \frac{a_{\ell-1} - 1}{b_{\ell-1} - 1} \\
&\geq - \frac{a_{\ell-1}}{b_{\ell-1}} \\
&\geq - \frac{k}{n}
\end{align*}

where the last inequality follows from Lemma~\ref{lem:helper}.
\end{enumerate}

In both cases, $Y_\ell > -1$ since $k \leq \frac{n}{2}$, thus completing the proof.
\end{proof}

Lemma~\ref{lem:martingale-diff} precisely computes how $X_\ell$ behaves as $\ell$ increases. If it is revealed that $m_\ell$ matches $Z_\ell$ to $i_\ell < Z_\ell$ (thus within $S$), then $X_\ell$ increases by some amount in the interval $[1 - \frac{2k}{n}, 1]$. Otherwise $X_\ell$ decreases by an amount in $[-\frac{k}{n}, 0]$. Using this enables us to bound the quadratic characteristic, and understand how the variance of $e(S)$ accumulates as subsequent $Z_\ell$ are conditioned on. 

\begin{lemma}\label{lem:martingale-quad-char}
Let $H \sim \gReg{n}{d}$ be a random regular graph, $S \subseteq V$ such that $\lvert S \rvert = k < \frac{n}{2}$, $N = d \cdot (k - 1)$. For $(X_\ell)_{\ell=0}^N$, we have $\langle X \rangle_N \leq \frac{k(k-1) d}{n - 2k}$ with probability 1.
\end{lemma}
\begin{proof}
It is sufficient to demonstrate $\langle X \rangle_{\ell} - \langle X \rangle_{\ell-1} \leq \frac{k}{n-2k}$ for all $\ell \in [N]$ as we would have
\begin{equation*}
\langle X \rangle_N
= \sum_{\ell=2}^N \big( \langle X \rangle_{\ell} - \langle X \rangle_{\ell-1} \big)
\leq N \cdot \frac{k}{n-2k}
\leq \frac{k(k-1)d}{n-2k}
\end{equation*}

Assume without loss of generality that $d = 1$ and fix $\ell$ along with a valid realization $Z_1 = z_1, \ldots, Z_{\ell-1} = z_{\ell-1}$. One can calculate the following fact
\begin{equation*}
\langle X \rangle_\ell - \langle X \rangle_{\ell - 1}
= \Var [ \bone \{ Z_\ell \in [k] \textup{ and } Z_\ell > \ell \} = 1 ]
\end{equation*}

Denote the indicator random variable $W_\ell = \bone \{ Z_\ell \in [k] \textup{ and } Z_\ell > \ell \}$. To bound the variance of the indicator, we seek to determine $\Pr [W_\ell = 1]$ with randomness taken over choice of $Z_\ell$. Recall that
\begin{equation*}
Y_\ell
= W_\ell + \big( \phi(a_\ell, b_\ell) - \phi(a_{\ell-1}, b_{\ell-1}) \big)
\end{equation*}

Note $Y_\ell$ is a random quantity since $W_\ell$, $a_\ell = a_\ell(z_1, \ldots, z_{\ell-1}, Z_\ell)$, and $b_\ell = b_\ell(z_1, \ldots, z_{\ell-1}, Z_\ell)$ each depend on a sample $Z_\ell$. It remains however that $\E [Y_\ell] = 0$ implying
\begin{equation*}
0 = \Pr \big[ W_\ell = 1 \big] + \E \big[ \big( \phi(a_\ell, b_\ell) - \phi(a_{\ell-1}, b_{\ell-1}) \big) \big]
\end{equation*}

and hence
\begin{equation*}
\Pr \big[ W_\ell = 1 \big]
= \E \big[ \phi(a_{\ell-1}, b_{\ell-1}) - \phi(a_\ell, b_\ell) \big]
\end{equation*}

Let us condition the expectation as follows.
\begin{align*}
\Pr \big[ W_\ell = 1 \big]
&= \Pr[W_\ell = 0] \cdot \E \big[ \phi(a_{\ell-1}, b_{\ell-1}) - \phi(a_\ell, b_\ell) \; \vert \; W_\ell = 0 \big] \\
&\qquad\qquad + \Pr[W_\ell = 1] \cdot \E \big[ \phi(a_{\ell-1}, b_{\ell-1}) - \phi(a_\ell, b_\ell) \; \vert \; W_\ell = 1 \big] \\
&\leq \E \big[ \phi(a_{\ell-1}, b_{\ell-1}) - \phi(a_\ell, b_\ell) \; \vert \; W_\ell = 0 \big] \\
&\qquad\qquad + \Pr[W_\ell = 1] \cdot \E \big[ \phi(a_{\ell-1}, b_{\ell-1}) - \phi(a_\ell, b_\ell) \; \vert \; W_\ell = 1 \big]
\end{align*}

Implying
\begin{equation*}
\Pr \big[ W_\ell = 1 \big]
\leq \frac{\E \big[ \phi(a_{\ell-1}, b_{\ell-1}) - \phi(a_\ell, b_\ell) \; \vert \; W_\ell = 0 \big]}{1 - \E \big[ \phi(a_{\ell-1}, b_{\ell-1}) - \phi(a_\ell, b_\ell) \; \vert \; W_\ell = 1 \big]}
\end{equation*}

Recall from the proof of Lemma~\ref{lem:martingale-diff} that $Y_\ell \geq 1 - \frac{2k}{n}$ if $W_\ell = 1$, while $Y_\ell \geq - \frac{k}{n}$ if $W_\ell = 0$. This means
\begin{align*}
&\E \big[ \phi(a_{\ell-1}, b_{\ell-1}) - \phi(a_\ell, b_\ell) \; \vert \; W_\ell = 0 \big] \leq \frac{k}{n} \\
&\E \big[ \phi(a_{\ell-1}, b_{\ell-1}) - \phi(a_\ell, b_\ell) \; \vert \; W_\ell = 1 \big] \geq \frac{2k}{n}
\end{align*}

and thus we have
\begin{equation*}
\Pr \big[ W_\ell = 1 \big] \leq \frac{k}{n - 2k}
\end{equation*}

Finally, as $W_\ell$ is an indicator random variable, its variance is at most that given by a Bernoulli random variable with success probability $\frac{k}{n - 2k}$. We conclude with
\begin{equation*}
\langle X \rangle_\ell - \langle X \rangle_{\ell - 1}
= \Var [ W_\ell ]
\leq \Pr [ W_\ell = 1 ]
\leq \frac{k}{n - 2k}
\end{equation*}

as required.
\end{proof}


\subsection{Concentration Analysis}

We now determine how $(X_\ell)$ concentrates. In~\cite{fan2012hoeffding}, the following Azuma-like inequality is proven for martingales.

\begin{theorem}[Remark 2.1 combined with equations (11) and (13) of \cite{fan2012hoeffding}]\label{thm:azuma-hoeffding}
Let $(X_\ell)_{\ell=0}^N$ be a martingale with martingale differences $(Y_\ell)$ satisfying $|Y_\ell| \leq 1$ for all $0 \leq \ell \leq N$. For every $0 \leq x \leq N$ and $\nu \geq 0$, we have
\begin{equation*}
\Pr \Big[ \lvert X_N - X_0 \rvert \geq x \textup{ and } \langle X \rangle_N \leq \nu^2 \Big]
\leq 2 \cdot \bigg( \frac{\nu^2}{x + \nu^2} \bigg)^{x + \nu^2} e^x
\end{equation*}
\end{theorem}

The concentration inequalities of~\cite{fan2012hoeffding} are one-sided inequalities as they are stated for supermartingales. We use the double-sided version, incurring an additional factor of 2 after taking a union bound with the negative of $(X_\ell)$. We start with a generic application of Theorem~\ref{thm:azuma-hoeffding} to fit our setting.

\begin{lemma}\label{lem:concentration-generic}
For $H$ a random regular graph drawn from $\gReg{n}{d}$, $S \subseteq V$ such that $\lvert S \rvert = k < \frac{n}{2}$, and $\delta > 0$, we have the following. 
\begin{equation*}
\Pr \Big[ \big\lvert e(S) - \E [e(S)] \big\rvert \geq \delta \cdot \E[e(S)] \Big]
\leq 2 \exp \bigg\{ - \E[e(S)] \cdot \bigg[ (\delta + C) \cdot \ln \bigg( \frac{\delta}{C} + 1 \bigg) - \delta \bigg] \bigg\}
\end{equation*}
where $C = \frac{2(n-1)}{n-2k}$
\end{lemma}
\begin{proof}
Let $x$ and $\nu^2$ be given by the following.
\begin{align*}
x 
&= \delta \cdot \E[e(S)]
= \delta \cdot \binom{k}{2} \frac{d}{n-1} \\
\nu^2
&= \frac{k(k-1)d}{n-2k}
= \binom{k}{2} \frac{d}{n-1} \cdot \frac{2(n-1)}{n-2k}
\end{align*}

By Lemma~\ref{lem:martingale-quad-char}, we have that $\langle X \rangle_N \leq \nu^2$ with probability one. Hence
\begin{equation*}
\Pr \Big[ \lvert X_0 - X_N \rvert \geq x \textup{ and } \langle X \rangle_n \leq \nu^2 \Big]
= \Pr \Big[ \lvert X_0 - X_N \rvert \geq x \Big]
= \Pr \Big[ \big\lvert e(S) - \E [e(S)] \big\rvert \geq \delta \E[e(S)] \Big]
\end{equation*}

Applying Theorem~\ref{thm:azuma-hoeffding} for the choice of $x, \nu^2$ above then concludes with the required bound.
\end{proof}

As mentioned previously, the purpose of choosing to study edges contained entirely in a set $S$ is because the number of edges contained entirely within $S$ can be written as a sum of fewer indicator random variables than the number of edges crossing the cut $(S, V-S)$. The difference between $\binom{k}{2}$ and $k(n - k)$ is not negligible (in particular for $k$ small) and we take advantage of this by further splitting our analysis of small cuts depending on the size of $k$.

To put this into broader context, we eventually apply Lemma~\ref{lem:concentration-generic} with choice of $\delta = \big( \frac{n}{k} - 2 \big) \cdot \frac{1.5}{\sqrt{d}}$ and $C = \frac{2(n-1)}{n-2k}$. A critical point here is that one can subsequently apply tighter approximations of the exponentiated term in Lemma~\ref{lem:concentration-generic} depending on the size of $k$, or more precisely, the size of $\frac{\delta}{C}$ which grows approximately as $\frac{n}{k \sqrt{d}}$. When $k \geq \Omega(n / \sqrt{d})$, applying the following Lemma~\ref{lem:log-taylor} yields tighter concentration.

\begin{restatable}{lemma}{logtaylor}\label{lem:log-taylor}
For any $\delta, C \geq 0$ such that $\frac{\delta}{C} \leq 1$, we have that
\begin{equation*}
\big( \delta + C \big) \ln \bigg( \frac{\delta}{C} + 1 \bigg)
\geq \delta + \frac{\delta^2}{3C}
\end{equation*}
\end{restatable}
Meanwhile, it is better to approximate the exponent using Lemma~\ref{lem:log-calc} below when $k \leq O(n / \sqrt{d})$. 

\begin{restatable}{lemma}{logcalc}\label{lem:log-calc}
For any $\delta \geq C \geq 1$, we have
\begin{equation*}
\big( \delta + C \big) \ln \bigg( \frac{\delta}{C} + 1 \bigg) - \delta 
\geq \frac{1}{2C} \cdot \delta \ln \delta
\end{equation*}
\end{restatable}

The proofs of Lemmas~\ref{lem:log-taylor} and~\ref{lem:log-calc} can be found in the \autoref{sec.analytic}. We additionally remark that though we study the number of edges contained entirely in $S$, justifying that $H$ cut sparsifies $G$ still requires computing the deviation of the number of edges crossing $(S, V-S)$. Scaling between edges contained within $S$ and crossing $(S, V-S)$ will thus explain the $\frac{n}{k}-2$ factor that appears in our choice of $\delta$. Let us now summarize the concentration bounds we use in each case via the following lemma.

\begin{lemma}\label{lem:concentration-cases}
For all sufficiently large choice of $n \geq 0$ and $d \geq 0$ constant such that given a random draw $H \sim \gReg{n}{d}$ and any $S \subset V$ such that $\lvert S \rvert = k$ where $2 \leq k \leq \frac{n}{100}$, the following statements hold 
\begin{enumerate}
\item If $\frac{\delta}{C} < 1$, then
\begin{equation}\label{eq.small-set-large-k}
\Pr \Big[ \big\lvert e(S) - \E [e(S)] \big\rvert \geq \delta \cdot \E[e(S)] \Big]
\leq 2 \exp\bigg( -\frac{2}{25} \cdot \frac{k^2d}{n} \cdot \delta^2 \bigg)
\end{equation}

\item If $\frac{\delta}{C} \geq 1$, then
\begin{equation}\label{eq.small-set-small-k}
\Pr \Big[ \big\lvert e(S) - \E [e(S)] \big\rvert \geq \delta \cdot \E[e(S)] \Big]
\leq 2 \exp\bigg( -\frac{49}{800} \cdot \frac{k^2d}{n} \cdot \delta \ln \delta \bigg)
\end{equation}
\end{enumerate}
where $C = \frac{2(n-1)}{n-2k}$
\end{lemma}
\begin{proof}
When $\frac{\delta}{C} < 1$, we can apply Lemma~\ref{lem:log-taylor} to approximate $(\delta + C) \cdot \ln \big( \frac{\delta}{C} + 1 \big)$ in Lemma~\ref{lem:concentration-generic} as follows.
\begin{align*}
\Pr \Big[ \big\lvert e(S) - \E [e(S)] \big\rvert \geq \delta \E[e(S)] \Big]
&\leq 2 \exp \bigg\{ - \E[e(S)] \cdot \bigg[ (\delta + C) \cdot \ln \bigg( \frac{\delta}{C} + 1 \bigg) - \delta \bigg] \bigg\} \\
&\leq 2 \exp \bigg\{ - \E[e(S)] \cdot \bigg[ \delta + \frac{\delta^2}{3C} - \delta \bigg] \bigg\} \\
&= 2 \exp \bigg\{ - \E[e(S)] \cdot \frac{\delta^2}{3C} \bigg\}
\end{align*}

Expanding $C$ and the expectation, we derive
\begin{align*}
\exp \bigg\{ - \E[e(S)] \cdot \frac{\delta^2}{3C} \bigg\}
&= \exp \bigg\{ - \binom{k}{2} \cdot \frac{d}{n-1} \cdot \frac{\delta^2}{3} \cdot \frac{n-2k}{2(n-1)} \bigg\} \\
&= \exp \bigg\{ - \frac{k^2d}{n} \cdot \delta^2 \cdot \frac{1}{12} \cdot \bigg( 1 - \frac{1}{k} \bigg) \cdot \bigg( 1 + \frac{1}{n-1} \bigg)^2 \cdot \bigg( 1 - \frac{2k}{n} \bigg) \bigg\} 
\end{align*}

Noticing that with large enough $n$, and as $k \leq \frac{n}{100}$, we will have that 
\begin{align*}
&\exp \bigg\{ - \frac{k^2d}{n} \cdot \delta^2 \cdot \frac{1}{12} \cdot \bigg( 1 - \frac{1}{k} \bigg) \cdot \bigg( 1 + \frac{1}{n-1} \bigg)^2 \cdot \bigg( 1 - \frac{2k}{n} \bigg) \bigg\} \\
&\leq \exp \bigg\{ - \frac{k^2d}{n} \cdot \delta^2 \cdot 0.99 \cdot \frac{1}{12} \cdot \frac{98}{100} \bigg\} \\
&\leq \exp\bigg( -\frac{2}{25} \cdot \frac{k^2d}{n} \cdot \delta^2 \bigg)
\end{align*}

as required. If $\frac{\delta}{C} \geq 1$, then we can apply Lemma~\ref{lem:log-calc} to approximate $(\delta + C) \cdot \ln \big( \frac{\delta}{C} + 1 \big)$ as follows 
\begin{equation*}
\Pr \Big[ \big\lvert e(S) - \E [e(S)] \big\rvert \geq \delta \E[e(S)] \Big]
\leq 2 \exp \bigg( - \frac{\E[e(S)]}{2C} \cdot \delta \ln \delta \bigg)
\end{equation*}

Expanding $C$ and the expectation, we derive 
\begin{align*}
\exp \bigg( - \frac{\E[e(S)]}{2C} \cdot \delta \ln \delta \bigg)
&= \exp \bigg\{ - \binom{k}{2} \cdot \frac{d}{n-1} \cdot \frac{n-2k}{2(n-1)} \cdot \delta \ln \delta \cdot \frac{1}{2} \bigg\} \\
&= \exp \bigg\{ - \frac{k^2d}{n} \cdot \delta \ln \delta \cdot \bigg( \frac{k-1}{k} \bigg) \cdot \bigg( \frac{n}{n-1} \bigg)^2 \cdot \bigg( 1 - \frac{2k}{n} \bigg) \cdot \frac{1}{8} \bigg\}
\end{align*}

Since $2 \leq k \leq \frac{n}{100}$, and for large enough $n$, we have that
\begin{align*}
\exp \bigg\{ - \frac{k^2d}{n} \cdot \delta \ln \delta \cdot \bigg( \frac{k-1}{k} \bigg) \cdot \bigg( \frac{n}{n-1} \bigg)^2 \cdot \bigg( 1 - \frac{2k}{n} \bigg) \cdot \frac{1}{8} \bigg\}
&\leq \exp \bigg\{ - \frac{k^2d}{n} \cdot \delta \ln \delta \cdot \frac{1}{2}  \cdot \frac{98}{100} \cdot \frac{1}{8} \bigg\} \\
&= \exp\bigg( -\frac{49}{800} \cdot \frac{k^2d}{n} \cdot \delta \ln \delta \bigg)
\end{align*}

as required.
\end{proof}

We now compute the probability that the number of edges contained within $S$ deviates far from its expectation. In the subsequent proof of Lemma~\ref{lem:small-cuts-contained}, the case of $\frac{\delta}{C} < 1$ is analogous to when $k \geq \Omega(n/\sqrt{d})$ while $\frac{\delta}{C} \geq 1$ corresponds to $k \leq O(n/\sqrt{d})$.

\begin{lemma}\label{lem:small-cuts-contained}
For all sufficiently large choice of $n \geq 0$ and $d \geq 0$ constant such that for $H \sim \gReg{n}{d}$ and any $S \subset V$ such that $\lvert S \rvert = k$ where $2 \leq k \leq \frac{n}{100}$, we have
\begin{equation*}
\Pr \Big[ \big\lvert e(S) - \E [e(S)] \big\rvert \geq \delta \E[e(S)] \Big]
\leq 2 \binom{n}{k}^{-1.01}
\end{equation*}
where $\delta = \big( \frac{n}{k} - 2 \big) \cdot \frac{1.5}{\sqrt{d}}$
\end{lemma}
\begin{proof}
With $C = \frac{2(n-1)}{n-2k}$, suppose $\frac{\delta}{C} < 1$, expanding $\delta$ in the bound given by equation~\eqref{eq.small-set-large-k}, we have
\begin{equation*}
\Pr \Big[ \big\lvert e(S) - \E [e(S)] \big\rvert \geq \delta \E[e(S)] \Big]
\leq 2 \exp\bigg( -\frac{2}{25} \cdot \frac{k^2d}{n} \cdot \delta^2 \bigg)
= 2 \exp\bigg\{ -\frac{2}{25} \cdot \frac{k^2d}{n} \cdot \bigg( \frac{n}{k} - 2 \bigg)^2 \cdot \frac{1.5^2}{d} \bigg\}
\end{equation*}

We now demonstrate how to upper bound this quantity by $\binom{n}{k}^{-1.01}$. It is equivalent to demonstrate
\begin{equation*}
\binom{n}{k}^{1.01} \leq \exp\bigg\{ \frac{2}{25} \cdot \frac{k^2d}{n} \cdot \bigg( \frac{n}{k} - 2 \bigg)^2 \cdot \frac{1.5^2}{d} \bigg\}
\end{equation*}

Taking the natural logarithm of both sides, and performing a change of variables $\alpha = \frac{k}{n}$, we have that 
\begin{equation*}
1.01 \cdot \ln \bigg( \binom{n}{\alpha n} \bigg)
\leq \frac{2 \cdot 1.5^2}{25} \cdot n \cdot \alpha^2 \bigg( \frac{1}{\alpha} - 2 \bigg)^2
\end{equation*}

As $\ln \big( \binom{n}{\alpha n} \big) \leq n \cdot H(\alpha)$ where $H$ denotes the binary entropy function, it is sufficient to demonstrate
\begin{equation*}
H(\alpha) \leq 0.18 \cdot (1 - 2 \alpha)^2
\end{equation*}

which holds for $\alpha \leq \frac{1}{100}$. Now suppose $\frac{\delta}{C} \geq 1$. Expanding $\delta$ in equation~\eqref{eq.small-set-small-k}, we have the following 
\begin{align*}
\Pr \Big[ \big\lvert e(S) - \E [e(S)] \big\rvert \geq \delta \E[e(S)] \Big]
&\leq 2 \exp\bigg( -\frac{49}{800} \cdot \frac{k^2d}{n} \cdot \delta \ln \delta \bigg) \\
&= 2 \exp\bigg\{ -\frac{49}{800} \cdot \frac{k^2d}{n} \cdot \bigg( \frac{n}{k} - 1 \bigg) \cdot \frac{1.5}{\sqrt{d}} \cdot \ln \bigg( \bigg( \frac{n}{k} - 2 \bigg) \cdot \frac{1.5}{\sqrt{d}} \bigg) \bigg\}
\end{align*}

Notice that since $k \leq \frac{n}{100}$, we have that $1 - \frac{2k}{n} \geq \frac{98}{100}$ meaning the expression can be upper bounded by
\begin{align*}
&\exp\bigg\{ -\frac{49}{800} \cdot \frac{k^2d}{n} \cdot \bigg( \frac{n}{k} - 1 \bigg) \cdot \frac{1.5}{\sqrt{d}} \cdot \ln \bigg( \bigg( \frac{n}{k} - 1 \bigg) \cdot \frac{1.5}{\sqrt{d}} \bigg) \bigg\} \\
&\leq \exp\bigg\{ -\frac{49}{800} \cdot \frac{98}{100} \cdot 1.5 \cdot k\sqrt{d} \cdot \ln \bigg( \frac{1.5n}{k\sqrt{d}} \cdot \frac{98}{100} \bigg) \bigg\}
\end{align*}

We next claim the following intermediate upper bound.
\begin{equation*}
\exp\bigg\{ -\frac{49}{800} \cdot \frac{98}{100} \cdot 1.5 \cdot k\sqrt{d} \cdot \ln \bigg( \frac{1.5n}{k\sqrt{d}} \cdot \frac{98}{100} \bigg) \bigg\}
\leq \binom{n}{\frac{k \sqrt{d} \cdot 100e}{1.5 \cdot 98 \cdot 150}}^{-1.01}
\end{equation*}

It is again equivalent to demonstrate the following
\begin{equation*}
1.01 \cdot \ln \bigg( \binom{n}{\frac{k \sqrt{d} \cdot 100e}{1.5 \cdot 98 \cdot 150}} \bigg)
\leq \frac{49}{800} \cdot \frac{98}{100} \cdot 1.5 \cdot k\sqrt{d} \cdot \ln \bigg( \frac{1.5n}{k\sqrt{d}} \cdot \frac{98}{100} \bigg)
\end{equation*}

However, because $\ln \binom{n}{k} \leq k \ln\big( \frac{en}{k} \big)$, it suffices to show that
\begin{equation*}
1.01 \cdot \frac{k\sqrt{d}}{1.5} \cdot \frac{100}{98} \cdot \frac{e}{150} \cdot \ln \bigg( 150 \cdot \frac{1.5n}{k\sqrt{d}} \cdot \frac{98}{100} \bigg)
\leq \frac{49}{800} \cdot \frac{98}{100} \cdot 1.5 \cdot k\sqrt{d} \cdot \ln \bigg( \frac{1.5n}{k\sqrt{d}} \cdot \frac{98}{100} \bigg)
\end{equation*}

which is equivalent to
\begin{equation*}
\ln \bigg( 150 \cdot \frac{1.5n}{k\sqrt{d}} \cdot \frac{98}{100} \bigg)
\leq \frac{49}{800} \cdot \bigg( \frac{98}{100} \bigg)^2 \cdot 1.5^2 \cdot \frac{150}{1.01 \cdot e} \cdot \ln \bigg( \frac{1.5n}{k\sqrt{d}} \cdot \frac{98}{100} \bigg)
\end{equation*}

Because $7$ lower bounds the constant on the right hand side, it is enough to show
\begin{equation*}
\ln \bigg( 150 \cdot \frac{1.5n}{k\sqrt{d}} \cdot \frac{98}{100} \bigg)
\leq 7 \cdot \ln \bigg( \frac{1.5n}{k\sqrt{d}} \cdot \frac{98}{100} \bigg)
\end{equation*}

As $\frac{1.5n}{k\sqrt{d}} \cdot \frac{98}{100} \geq \delta \geq C \geq 1$, for all large enough $n$, the above holds. Finally, we show
\begin{equation*}
\binom{n}{\frac{k \sqrt{d} \cdot 100e}{1.5 \cdot 98 \cdot 150}}^{-1.01}
\leq \binom{n}{k}^{-1.01}
\end{equation*}

by choosing a $d$ large enough since $\frac{\delta}{C} \geq 1$. A choice of $d \geq \big( \frac{1.5 \cdot 98 \cdot 150}{100e} \big)^2$ suffices.
\end{proof}


\subsection{Completing the Proof}

Finishing the analysis of the small cuts regime, we now show the main result stated at the beginning of this section: the number of edges crossing $(S, V-S)$ deviates no more from its expectation than by a $\frac{1.5}{\sqrt{d}}$ factor with high probability.

\smallcuts*

\begin{proof}
Denote $\lvert S \rvert = k$. If $k=1$, then $\cut_H(S) = d$ for any random $d$ regular graph $H$ thus $\cut_H(S) - \E [\cut_H(S)] = 0$. Now consider any $2 \leq k \leq \frac{n}{100}$. Because $H$ is $d$ regular, we have 
\begin{equation*}
\cut_H(S) = kd - 2 \cdot e_H(S) 
\end{equation*}

Thus the event $\{ \lvert \cut_H(S) - \E [\cut_H(S)] \rvert \geq \frac{1.5}{\sqrt{d}} \cdot \E[\cut_H(S)] \}$ occurs if and only if
\begin{align*}
\big\lvert \cut_H(S) - \E [\cut_H(S)] \big\rvert 
&\geq \frac{1.5}{\sqrt{d}} \cdot \E[\cut_H(S)] \\
\big\lvert kd - 2 \cdot e_H(S) - \E [kd - 2 \cdot e_H(S)] \big\rvert 
&\geq \frac{1.5}{\sqrt{d}} \cdot \E[kd - 2 \cdot e_H(S)] \\
\big\lvert \E [e_H(S)] - e_H(S) \big\rvert 
&\geq \frac{1.5}{\sqrt{d}} \cdot \E[e_H(S)] \cdot \bigg( \frac{kd}{2 \E[e_H(S)]} - 1 \bigg) \\
\big\lvert e_H(S) - \E [e_H(S)] \big\rvert 
&\geq \frac{1.5}{\sqrt{d}} \cdot \E[e_H(S)] \cdot \bigg( \frac{n-1}{k-1} - 1 \bigg)
\end{align*}

Now, $\frac{n-1}{k-1} - 1 = \frac{n}{k} - \big( 1 + \frac{1}{k-1} \big) \geq \frac{n}{k} - 2$ since $k \geq 2$. The probability of the above occurring is at most 
\begin{equation*}
\Pr \bigg[ \big\lvert \cut_H(S) - \E[\cut_H(S)] \big\lvert \geq \frac{1.5}{\sqrt{d}} \cdot \E[\cut_H(S)] \bigg]
\leq \Pr \bigg[ \big\lvert e_H(S) - \E[e_H(S)] \big\lvert \geq \frac{1.5}{\sqrt{d}} \cdot \bigg( \frac{n}{k} - 2 \bigg) \cdot \E[e_H(S)] \bigg]
\end{equation*}

Applying Lemma~\ref{lem:small-cuts-contained} using $\delta = \frac{1.5}{\sqrt{d}} \cdot \big( \frac{n}{k} - 2 \big)$ implies that the right hand side is at most $o_{n}\big( \binom{n}{k}^{-1} \big)$. Performing a union bound over at most $\binom{n}{k}$ cuts of size $k$ then completes the proof.
\end{proof}

\section{Lower bounds for spectral sparsification of the clique}
\label{sec:lower-bound}

In the following, if $H=(V,E_H,w_H)$ is an undirected weighted graph and $v\in V$ is a vertex, we call the {\em combinatorial degree} of $v$ the number of edges incident on $v$, and we call the {\em weighted degree} of $v$ the sum of the weights of the edges incident on $v$. A random walk in a graph is a process in which we move among the vertices of a graph and, at every step, we move from the current node $u$ to a neighbor $v$ of the $u$ with probability proportional to the weight of the edge $(u,v)$. We will denote the \emph{complete graph} on $n$ vertices with each edge weighted $1/(n-1)$ as $\bar K_n$. The \emph{complete bipartite graph} on $n$ vertices with equal sized partitions, and each edge weighted $2/n$ is denoted as $\bar K_{n/2,n/2}$.

A {\em symmetric doubly stochastic} matrix $M$ is a non-negative matrix whose rows and columns sum to 1. In this case, the all-ones vector $\ones = (1, 1, \ldots, 1)$ is an eigenvector with eigenvalue 1, which we think of as as the {\em trivial eigenvalue} of $M$. To capture our lower bounds for sparsifiers exhibiting large ``odd pseudo-girth,'' we require considering a relaxation of double stochasticity; a matrix $M$ is \emph{$\epsilon$-almost doubly stochastic} if $M$ is a non-negative square matrix whose row- and column-sums are between $1-\epsilon$ and $1+\epsilon$.

In this section, we prove the following result.

\begin{restatable}[Lower Bound for Spectral Sparsification]{theorem}{lowerboundspectral}
\label{thm:lower-bound-spectral}
Let $H=(V,E,w)$ be a weighted graph on $n$ vertices and with $dn/2$ edges, so that $H$ has average combinatorial degree $d$.
If $H$ satisfies any of the conditions
\begin{enumerate}
    \item (weighted regular) $H$ is an $\epsilon$\! spectral sparsifier of $\bar K_n$ and its adjacency matrix is doubly stochastic,
    \item (large odd pseudo-girth) $H$ is an $\epsilon$\! spectral sparsifier of $\bar K_n$ and at most $B$ vertices of $H$ participate in odd cycles of length at most $g$ such that $B = o(n)$ and $g \leq d^{1/4}$,
    \item (bipartite) $H$ is an $\epsilon$\! spectral sparsifier of $\bar K_{n/2,n/2}$ and $H$ is bipartite.
\end{enumerate}
then $\epsilon$ must also satisfy the following.
\[
\epsilon
\ge \frac{2}{\sqrt{d}} - O \bigg( \frac{\ln d}{d^{3/4}} \bigg)  - O \bigg( \frac{d^{\sqrt{d}/(\ln d)^2}}{n\sqrt{d}} \bigg) - O \bigg( \frac{1}{\sqrt{n}} \bigg)
\]
\end{restatable}

The three classes of graph sparsifers named in the theorem come from a combination of three properties that make it convenient to apply our techniques: (1) they are almost-doubly-stochastic (2) the odd powers of their adjacency matrices have non-negative or bounded traces (3) the graph being sparsified has almost all of its eigenvalues equal (or nearly equal) to zero.

The proof of \autoref{thm:lower-bound-spectral} requires a lower bound on the trace of symmetric, (almost) doubly stochastic matrices, which also gives our ``Alon-Boppana theorem'' for symmetric, doubly stochastic matrices as a corollary.

\doublystoc*

Note that case (2) in \autoref{thm:lower-bound-spectral} implies the statement of \autoref{th.speclower} from the introduction by taking $B = \sqrt{n}$ and $g = d^{1/4}$.

\speclower*


\subsection{Outline of the Proof}
\label{sec.lboutline}

The broad approach is to lower-bound the trace of a high power of the adjacency matrix, then convert such a trace lower bound into a explicit construction of a test vector which distinguishes between the Laplacian of the sparsifier and the Laplacian of the graph being sparsified.

\subsubsection{The Trace Lower Bound}

We first establish a lower bound to the trace of a sparse (almost) doubly stochastic matrix in \autoref{sec.trace}. Our argument modifies that presented in~\cite{hoory2005lower} to handle symmetric doubly stochastic transition matrices rather than the (non-symmetric) transition matrices of random walks on irregular graphs. When $M$ is strictly doubly stochastic, our lower bound can be stated as follows.
\begin{lemma}
\label{lem:doubly-stochastic}
Let $M \in \R^{n \times n}$ be a symmetric doubly stochastic matrix with $dn$ non-zero entries, such that $M_{a,b} \le 2/\sqrt{d}$ for all $a,b \in [n]$. Then for each $k \le d^{1/4}$,
\[ \trace M^{2k} \ge n\left(\frac{2 - O(\frac{\ln d}{k})}{\sqrt{d}}\right)^{2k}. \]
\end{lemma}

The proof of \autoref{lem:doubly-stochastic} is provided in \autoref{sec.trace}. A slight modification of the argument yields a similar bound for non-negative symmetric matrices whose row- and column-sums are bounded within $1 \pm 2/\sqrt{d}$.
\begin{lemma}
\label{lem:almost-doubly-stochastic}
Let $M \in \R^{n \times n}$ be a symmetric $\frac{2}{\sqrt{d}}$-almost doubly stochastic matrix with $dn$ non-zero entries, such that $M_{a,b} \le 2/\sqrt{d}$ for all $a,b \in [n]$.
Then for each $k \le d^{1/4}$,
\[ \trace M^{2k} \ge n\left(\frac{2 - O(\frac{\ln d}{k})}{\sqrt{d}}\right)^{2k}. \]
\end{lemma}

The proof does not present a significant deviation from the techniques presented in \autoref{lem:doubly-stochastic} hence we relegate it to \autoref{sec.almost-doubly-stoc}. The argument of \autoref{lem:doubly-stochastic} involves expressing the trace of the $2k$-th power of a transition matrix as the sum of probabilities of closed walks of length $2k$.
One subset of the closed walks are those formed by $k$ steps of a \emph{non-backtracking walk}. This is a walk $w_0,\dots,w_k \in [n]$ where $w_{i-1} \ne w_{i+1}$ for all $i\in[k-1]$, followed by $k$ steps of the exact reverse walk.
Summing the probabilities of that particular \emph{shape} of walk gives a lower bound of $n/d^k$ for the trace of the $2k$-th power of $M$, which falls short by a factor of $4^k$.
To recover that factor of $4^k$, we consider also other shapes of closed walks, so we will take an interlude in \autoref{sec:walks} to introduce these ``shapes'' along with requisite notation.

\subsubsection{Test Vector}

The condition that $H$ is an $\epsilon$\! spectral sparsifier of $G$ can be written as
\begin{equation*}
\forall x \in \R^V \ \ \
(1-\epsilon) x^{\T} L_{G} x
\leq x^{\T} L_{H} x
\leq (1+\epsilon) x^{\T} L_{G} x
\end{equation*}
or equivalently
\begin{equation*}
\forall x \in \R^V \ \ \
(1-\epsilon) \iprod{ xx^{\T}, L_{G} }
\leq \iprod{ xx^{\T}, L_{H} }
\leq (1+\epsilon) \iprod{ xx^{\T}, L_{G} }
\end{equation*}
where $\iprod{A, B} = \sum_{i,j} A_{i,j} B_{i,j}$ denotes the Frobenius inner product between real-valued, square matrices. Because positive semidefinite matrices $X \succeq \mathbf{0}$ are convex combinations of rank-1 symmetric matrices of the form $xx^{\T}$, we can write the condition that $H$ $\epsilon$\! spectrally sparsifies $G$ as
\begin{equation*}
\forall X \succeq \mathbf{0} \ \ \
(1-\epsilon) \iprod{ X, L_{G} }
\leq \iprod{ X, L_{H} }
\leq (1+\epsilon) \iprod{ X, L_{G} }
\end{equation*}

Hence to prove \autoref{thm:lower-bound-spectral}, we look for a positive semidefinite matrix $X \succeq \mathbf{0}$ for which $\iprod{ X, L_{G} }$ is noticeably different from $\iprod{ X, L_{H} }$. This approach is equivalent to the approach of considering probability distributions over test vectors $x$ which is taken in \cite{ST18}.

Our lower bound in \autoref{lem:almost-doubly-stochastic} does not immediately imply lower bounds on the approximation error for sparsifiers with almost doubly stochastic adjacency matrices because lower bounding error density for spectral sparsification requires lower bounding non-trivial eigenvalues of the difference between $L_H$ and $L_G$. Thus, we use \autoref{lem.almostdoublystoc}, proven in \autoref{sec.matrix}, to establish a link between non-trivial eigenvalues of $A_H$ and $L_H - L_G$. This lemma demonstrates the existence of an explicit matrix $X$ with $M - D + I$ where $M$ is almost doubly stochastic, and $D$ is the diagonal matrix of $M$'s row sums.

\begin{lemma}
\label{lem.almostdoublystoc}
Let $M$ be an $n\times n$ symmetric $2/\sqrt{d}$-almost doubly stochastic matrix with $dn$ non-zero entries, and $D$ be the diagonal matrix where $D_{i,i} = \sum_{j \in [n]}M_{i,j}$ for all $i \in [n]$. Suppose there exists $\gamma > 0$, such that for all $k \leq d^{1/4}$, $M$ satisfies the following.
\begin{equation*}
\frac{ \iprod{M^{2k-1},D-I}}{\trace M^{2k}} \le \frac{\sqrt{d}}{n^{\gamma}}
\end{equation*}
Then for any $0 < \delta \le 1/2$, if $S \subseteq \R^n$ is a subspace with $\dim S \le n^{\delta}$ and $\lVert v \rVert_{\infty} / \lVert v \rVert_2 \leq 1 / \big( n^{1/4 + \delta/2}\big)$ for all $v \in S$, there exists an $X \succeq \mathbf{0}$ satisfying $S \subseteq \ker X$ so that
\[
\frac{|\iprod{X, M-D+I}|}{\trace X}
\ge \frac{2}{\sqrt d} - O \bigg( \frac{\ln d}{d^{3/4}} \bigg) - O\bigg( \frac{d^{\sqrt{d} / (\ln d)^2}}{n^{1-\delta} \cdot \sqrt{d}}\bigg) - O\left(\frac{1}{n^{\delta}}\right) - O\left(\frac{1}{n^{\gamma}}\right).
\]
\end{lemma}

This immediately gives \autoref{th.speclower} since $M - D + I = M$ when $M$ is strictly doubly stochastic. The proof of \autoref{thm:lower-bound-spectral} will note that when $M = A_H$, we have $D = D_H$, and $\iprod{X, M - D + I} = \iprod{X, L_H - L_G}$ for all $X \succeq \mathbf{0}$ such that $\ker A_G \subseteq \ker X$ where $G$ is either $\bar K_n$ or $\bar K_{n/2, n/2}$. We give this proof at the end, in \autoref{sec.final-sep}.

\subsubsection{Technical Caveats}

Towards proving \autoref{thm:lower-bound-spectral}, we assume without loss of generality that the weighted degree of any vertex in $H$ is within $1 \pm 2/\sqrt{d}$. If this does not hold then a simpler proof leads to the conclusion of \autoref{thm:lower-bound-spectral}.

\begin{restatable}{lemma}{smallwtddegree}
\label{lem.doublystoc-smallwtddegree}
Let $H=(V,E,w)$ be a weighted graph on $n$ vertices.
Let $\bar K_n$ be a clique on $V$ with every edge weighted $1/(n-1)$.
Let $\bar K_{n/2,n/2}$ be a complete bipartite graph on $V$ with every edge weighted $2/n$.
Suppose $H$ is an $\epsilon$\! spectral sparsifier of either $\bar K_n$ or $\bar K_{n/2, n/2}$. If there exists a vertex such that its weighted degree is either larger than $1 + 2/\sqrt{d}$ or smaller than $1 - 2/\sqrt{d}$, then
\begin{equation*}
\epsilon \geq \frac{2}{\sqrt{d}}
\end{equation*}
\end{restatable}

This is similar to assumptions made by~\cite{ST18} on the structure of the sparsifier, except that ours additionally apply to sparsifying the complete bipartite graph. This also motivates why we choose to analyze $2/\sqrt{d}$-almost doubly stochastic matrices. If each vertex of a certain graph has weighted degree bounded by $1 \pm 2/\sqrt{d}$, then its weighted adjacency matrix is also $2/\sqrt{d}$-almost doubly stochastic. The proof of \autoref{lem.doublystoc-smallwtddegree} is given in \autoref{sec.sparsifier-assumptions}.

A technical caveat required by the proof of \autoref{lem.almostdoublystoc} is that one can assume that the entries of $M_{a,b} \leq \frac{2}{\sqrt{d}}$. Otherwise, a simpler choice of test vector yields the conclusion of \autoref{lem.almostdoublystoc}. Its proof is also in \autoref{sec.sparsifier-assumptions}.

\begin{restatable}{lemma}{smalledgewt}
\label{lem.doublystoc-smalledgewt}
Let $M$ be an $n\times n$ symmetric $2/\sqrt{d}$-almost doubly stochastic matrix, and $D$ be the diagonal matrix where $D_{i,i} = \sum_{j \in [n]}M_{i,j}$ for all $i \in [n]$. For any $0 < \delta \le 1/2$, let $S \subseteq \R^n$ be a subspace with $\dim S \le n^{\delta}$ and $\lVert v \rVert_{\infty} / \lVert v \rVert_2 \leq 1 / \big( n^{1/4 + \delta/2}\big)$ for all $v \in S$. If there are $a, b \in [n]$ such that $M_{a, b} \geq \frac{2}{\sqrt{d}}$, then there is an $X \succeq \mathbf{0}$ satisfying $S \subseteq \ker X$ so that
\[ \frac{\iprod{X,M-D+I}}{\trace X} \ge \frac{2}{\sqrt{d}} -  O \bigg( \frac{1}{n^\delta} \bigg). \]
\end{restatable}


\subsection{Notation for walks}
\label{sec:walks}

\begin{figure}
    \centering
    \begin{tikzpicture}[y=.7cm, x=.7cm,font=\sffamily,scale=1.25, every node/.style={scale=1.25}]
    \tikzset{vertex/.style={circle,fill=white,draw,inner sep=0pt,outer sep=0pt,minimum width=1.5em}}
    \tikzset{decrea/.style={regular polygon,regular polygon sides=4,fill=white,draw=red,inner sep=0pt,outer sep=0pt,minimum width=1.8em}}
	\draw (0,0) -- coordinate (x axis mid) (6,0);
    	\draw (0,0) -- coordinate (y axis mid) (0,2.5);
    	\foreach \x in {0,...,6}
     		\draw (\x,1pt) -- (\x,-5pt)
			node[anchor=north,below=0.1cm] {\x};
    	\foreach \y in {0,1,...,2}
     		\draw (1pt,\y) -- (-5pt,\y)
     			node[anchor=east,left=0.1cm] {\y};
	\node[below=0.7cm] at (x axis mid) {$i$};
	\node[above=1.2cm,left=0.0cm] at (y axis mid) {$\tau_i$};

	\node[vertex] at (0,0) (w0) {$\alpha$};
	\node[vertex] at (1,1) (w1) {$\beta$} edge [-] (w0);
	\node[vertex] at (2,2) (w2) {$\gamma$} edge [-] (w1);
	\node[decrea] at (3,1) (w3) {$\beta$} edge [-] (w2);
	\node[vertex] at (4,2) (w4) {$\delta$} edge [-] (w3);
	\node[decrea] at (5,1) (w5) {$\beta$} edge [-] (w4);
	\node[decrea] at (6,0) (w6) {$\alpha$} edge [-] (w5);

	\node[inner sep=0pt,outer sep=0pt] at (2.75,3.2) (labela) {\small $a(\tau,\omega,3)$} edge [->,shorten >=4pt,{>=stealth}] (w3);
	\node[inner sep=0pt,outer sep=0pt] at (5,3.5) (labelb) {\small $b(\tau,\omega,3)$} edge [->,shorten >=4pt,>=stealth] (w4);

	\node[vertex] at (8.55,1.3) (alpha) {$\alpha$};
	\node[vertex] at (10,1.3) (beta) {$\beta$} edge [-] (alpha);
	\node[vertex] at (11,2.3) (gamma) {$\gamma$} edge [-] (beta);
	\node[vertex] at (11,0.3) (delta) {$\delta$} edge [-] (beta);

    \end{tikzpicture}
    \caption{
    An example instantiation (left) of a Dyck path $\tau = [0, 1, 2, 1, 2, 1, 0]$ specified by the set $\omega = [\alpha,\beta,\gamma,\delta]$.
    Here the walk is $w = W(\tau,\omega) = [\alpha, \beta, \gamma, \beta, \delta, \beta, \alpha]$.
    Increasing steps and the $0$th step are drawn as circles while decreasing steps are portrayed as red squares.
    Note that the increasing and $0$th steps correspond to the components of $\omega$, and that the decreasing steps' instantiations are determined by looking to the left on the chart.
    The endpoints $a(\tau,\omega,3)$ and $b(\tau,\omega,3)$ of the third increasing step are labelled.
    This is a non-backtracking instantiation because $w_{2} = \gamma \ne \alpha = w_0$ and $w_4 = \delta \ne \alpha = w_0$.
    We may imagine $w$ as a walk over a graph (right) on the vertex set $\{\alpha,\beta,\gamma,\delta\}$. }
    \label{fig:walks}
\end{figure}
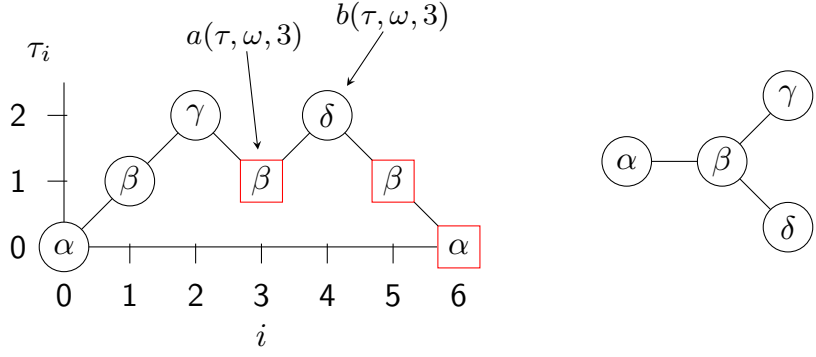

Let $\cT_k \subseteq \N^{2k+1}$ be the set of zero-indexed Dyck paths of length $2k$, so that if $\tau \in \cT_k$, then $\tau_0 = 0$ and $\tau_{2k} = 0$ and $\tau_i \ge 0$ for all $i$ and $|\tau_i - \tau_{i-1}| = 1$ for all $i \ge 1$.
An \emph{increasing step} of $\tau$ is an index $i$ such that $\tau_{i} = \tau_{i-1} + 1$ and a \emph{decreasing step} is $i$ where $\tau_{i} = \tau_{i-1} - 1$. The index $0$ is neither increasing nor decreasing.

We consider ways to \emph{instantiate} the Dyck path $\tau \in \cT_k$ as a walk over $[n]$.
These instantiations are specified by a zero-indexed vector $\omega$ in $[n]^{k+1}$.
The instantiation specified by $\omega$ is given by the walk
\[ W(\tau,\omega) = (w_0, \dots, w_{2k}) \in [n]^{2k+1}\]
so that $w_0 = \omega_0$, and if $i$ is the $j$th increasing step of $\tau$ then $w_{i} = \omega_j$,
and if $i$ is a decreasing step then $w_{i} = w_j$ where in this case, $j$ is the largest index such that $j < i$ and $\tau_{i} = \tau_j$.
An example instantiation is given in \autoref{fig:walks}.

So we interpret this as a walk that advances according to a stack of state transitions: for every increasing step of $\tau$, a new element of $[n]$ from the list $\omega$ is advanced to, and the corresponding transition is added to the stack.
For every decreasing step of $\tau$, we reverse the transition on the top of the stack.
In this way, each transition is taken an even number of times (once forwards and once backwards), and we always end up where we started.

In a \emph{non-backtracking instantiation} of $\tau$, we additionally have the requirement that $w_{i} \ne w_j$ for all increasing steps $i$ where $j$ (if it exists) is the largest index satisfying $j < i$ and $\tau_{j} = \tau_{i} - 2$.
Meaning that on a increasing step of $\tau$, we are not allowed to advance to the same state that we would retreat to if we were on a decreasing step instead.
Let $\Omega_{\tau}$ be the set of all $\omega$ so that $W(\tau,\omega)$ is a non-backtracking instantiation of $\tau$.

Let $a(\tau,\omega,j) \in [n]$ be the origin of the $j$th increasing step of $\tau$ as instantiated in $W(\tau,\omega)$, and let $b(\tau,\omega,j) \in [n]$ be its destination.
More precisely, if $i$ is the index of the $j$th increasing step of $\tau$, then
\[ a(\tau,\omega,j) = W(\tau,\omega)_{i-1} \text{\quad and \quad} b(\tau,\omega,j) = W(\tau,\omega)_{i}.\]

For $w \subseteq [n]^{2k+1}$, let
\[ p(M,w) = \prod_{i \in [2k]} M_{w_{i-1},w_{i}},\]
so that $p(M,w)$ is the probability. when starting at state $w_0$, of transitioning under the probability matrix $M$ to each of the successive states in $w$ in the next $2k$ transitions.

Note that
\[ p(M,W(\tau,\omega)) = \prod_{i \in [k]} M_{a(\tau,\omega,i),b(\tau,\omega,i)}^2\]
when $M$ is symmetric, and let
\[ q(M, \tau, \omega) = \prod_{i \in [k]} M_{a(\tau,\omega,i),b(\tau,\omega,i)},\]
so that $q(M,\tau,\omega)$ is the combined probability according to $M$ of each of the increasing steps of $\tau$ as instantiated by $\omega$, when conditioned on all of the deceasing steps occurring according to $\tau$.
Then we have the identity
\[ p(M,W(\tau,\omega)) = q(M, \tau, \omega)^2.\]


\subsection{The Trace Lower Bound}
\label{sec.trace}

Before proving \autoref{lem:doubly-stochastic}, we introduce a technical lemma which generalizes the idea that in a graph where all edge weights are small, most short walks (weighted by probability) are non-backtracking walks.

\begin{lemma}
\label{lem:Z}
Let $M \in \R^{n \times n}$ be a symmetric, doubly stochastic matrix such that no entry is greater than $2/\sqrt{d}$.
For $\tau \in \cT_k$, let
\[ Z = \frac{1}{n}\sum_{\omega \in \Omega_{\tau}} q(M, \tau, \omega).\]
Then
\[ Z \ge 1 - \frac{2k}{\sqrt{d}}. \]
\end{lemma}
\begin{proof}
Let $\Omega_{\tau,j} \subseteq [n]^{j+1}$ be the set of truncations of elements of $\Omega_{\tau}$ to lists of length $j+1$ instead of $k+1$.
Let $q(M,\tau,\omega,j) = \prod_{i \in [j]} M_{a(\tau,\omega,i),b(\tau,\omega,i)}$ be the corresponding truncation of $q(M,\tau,\omega)$.
We set up an induction with the inductive hypothesis that
\[ \frac{1}{n}\sum_{\omega \in \Omega_{\tau,j+1}} q(M, \tau, \omega, j+1) \ge 1-\frac{2j}{\sqrt{d}}. \]

In the base case where $j = 0$, there are $n$ elements of $\Omega_{\tau,1}$, and $q(M,\tau,\omega,0)$ is an empty product, so the sum is $n$.

In the inductive step, we partition $\Omega_{\tau,j+1}$ into sets $S_{\omega}$ for $\omega \in \Omega_{\tau,j}$, so that $S_{\omega}$ is the subset of $\Omega_{\tau,j+1}$ with $\omega$ as a prefix. Then for each $\omega$, there is a single $z_{\omega} \in [n]$ so that the concatenation $\omega \circ (z_{\omega})$ of $z$ onto $\omega$ is not the prefix of any $\omega'$ that would make $W(\tau,\omega')$ a non-backtracking walk, in other words $\omega \circ (z_{\omega}) \not\in S_{\omega}$.
So
\[
\sum_{s \in S_{\omega}} q(M, \tau, s, j+1) = \sum_{x \ne z_{\omega}} q(M,\tau,\omega,j) M_{a(\tau,\omega,j),x}
\]
Since $\sum_{x \in [n]} M_{y,x} = 1$ for all $y$ and the entries of $M$ are at most $2/\sqrt{d}$,
\[
\sum_{s \in S_{\omega}} q(M, \tau, s, j+1) = (1 - M_{a(\tau,\omega,j),z})\, q(M,\tau,\omega,j) \ge \left(1-\frac{2}{\sqrt{d}}\right) q(M,\tau,\omega,j)
\]
So since $\{S_{\omega}\}$ is a partition of $\Omega_{\tau,j+1}$,
\begin{align*}
\frac{1}{n}\sum_{\omega \in \Omega_{\tau,j+1}} q(M, \tau, \omega,j+1)
&= \frac{1}{n}\sum_{\omega \in \Omega_{\tau,j}} \sum_{s \in S_{\omega}} q(M, \tau, s, j+1)
\\&\ge \left(1-\frac{2}{\sqrt{d}}\right)\frac{1}{n}\sum_{\omega \in \Omega_{\tau,j}} q(M, \tau, \omega, j)
\\&\ge \left(1-\frac{2}{\sqrt{d}}\right)\left(1-\frac{2(j-1)}{\sqrt{d}}\right)
\\&\ge 1 - \frac{2j}{\sqrt{d}}. \qedhere
\end{align*}
\end{proof}

We are now ready to prove \autoref{lem:doubly-stochastic}.

\begin{proof}[Proof of \autoref{lem:doubly-stochastic}]
From the definition of trace,
\[ \trace M^{2k} = \sum_{v \in [n]} \;\sum_{\substack{w_0, \dots, w_{2k} \in [n] \\ w_0 = w_{2k} = v}} \; \prod_{j \in [2k]} M_{i_{j-1},i_j}
= \sum_{v \in [n]} \;\sum_{\substack{w_0, \dots, w_{2k} \in [n] \\ w_0 = w_{2k} = v}} \; p(M,w). \]
Each summand corresponds to the probability of a closed walk starting at $v$ and ending at $v$.
Each summand is also non-negative, and therefore, we can bound the sum by a subset of the walks:
\begin{equation}
\label{eq:catalan-subset}
    \frac{1}{n}\trace M^{2k} \ge \frac{1}{n}\sum_{\tau \in \cT_k} \sum_{\omega \in \Omega_{\tau}} p(M, W(\tau, \omega)).
\end{equation}

We would now like to understand the sum for each fixed $\tau \in \cT_k$.
Let
\[ S_{\tau} = \frac{1}{n}\sum_{\omega \in \Omega_{\tau}} p(M, W(\tau, \omega)) \]
and let $Z = \tfrac{1}{n}\sum_{\omega \in \Omega_{\tau}} q(M, \tau, \omega)$.
By the weighted arithmetic-mean--geometric-mean inequality with weights $\frac{1}{Zn}q(M,\tau,\omega)$ and the fact that $p(M,W(\tau,\omega)) = q(M, \tau, \omega)^2$,
\[
S_{\tau}
= Z \sum \left(\frac{1}{Zn}q(M,\tau,\omega)\right) \cdot q(M, \tau, \omega)
\ge
Z\prod_{\omega \in \Omega_{\tau}} q(M, \tau, \omega)^{q(M, \tau, \omega)/(nZ)}.
\]
By expanding the definition of $q(M, \tau, \omega)$,
\[S_{\tau} \ge Z\prod_{\omega \in \Omega_{\tau}} \prod_{i \in [k]} M_{a(\tau,\omega,i),b(\tau,\omega,i)}^{q(M, \tau, \omega)/(nZ)}.
\]
Letting $\#(e,\tau,\omega)$ be the number of different $i$ for which $e = (a(\tau,\omega,i),b(\tau,\omega,i))$,
\[S_{\tau} \ge Z\prod_{\omega \in \Omega_{\tau}} \prod_{e \in [n] \times [n]} M_{e}^{q(M, \tau, \omega)\#(e,\tau,\omega)/(nZ)}.
\]
We can lower bound the above by replacing the $\Omega_{\tau}$ in the first product by the larger set $[n]^{k+1}$, because every new factor that's added is at most $1$:
\[ S_{\tau} \ge Z\prod_{\omega \in [n]^{k+1}} \prod_{e \in [n] \times [n]} M_{e}^{q(M, \tau, \omega)\#(e,\tau,\omega)/(nZ)}.
\]
Thus, carrying the first product into the exponent,
\begin{equation}
\label{eq:time-for-stationarity}
S_{\tau} \ge
Z\prod_{e \in [n] \times [n]} M_{e}^{\sum_{\omega \in [n]^{k+1}} q(M, \tau, \omega)\#(e,\tau,\omega)/(nZ)}.
\end{equation}
Now we would like to understand the quantity in the exponent.

We'll denote the truncation of $q(M,\tau,\omega)$ to the product of the probabilities of the first $j$ increasing steps as
\[ q(M,\tau,\omega,j) = \prod_{i \in [j]} M_{a(\tau,\omega,i),b(\tau,\omega,i)}.\]
So that, letting $\Ind$ be the indicator function so that $\Ind[P] = 1$ if the proposition $P$ is true and $0$ otherwise,
\[ q(M, \tau, \omega)\#(e,\tau,\omega) = \sum_{j \in [k]} q(M,\tau,\omega,j)\Ind[e = (a(\tau,\omega,j),b(\tau,\omega,j))].\]
We characterize $q(M,\tau,\omega,j)$ as the probability of the first $j$ elements of $\omega$ being the first $j$ states of a process that starts at the state $\omega_0 \in [n]$ and then, at the $i$th step, either
\begin{itemize}
    \item advances to a new random state $v_i$ with probability $M_{v_{i-1},v_i}$ if $i$ is an increasing step of $\tau$, or
    \item if instead $i$ is a decreasing step of $\tau$, retreats to some predetermined previous state $v_{i'}$ where $i'$ is a function of $\tau$ and $i$ only, and not a function of $v$, and then advances to a new random state $v_i$ with probability $M_{v_{i'},v_i}$.
\end{itemize}
The distribution that assigns each state $v \in [n]$ equal probability $1/n$ is a stationary distribution at each step of this process: its stationarity for increasing steps is an immediate consequence of $M$ being doubly stochastic as a transition matrix, and \emph{any} distribution would be stationary for the steps where we jump back to a previous state.
The probability of this process taking a particular transition $e$ as its $j$th increasing step after starting in its stationary distribution is given by
\[\frac{1}{n}\sum_{\omega \in [n]^{k+1}} q(M,\tau,\omega,j)\Ind[e = (a(\tau,\omega,j),b(\tau,\omega,j))],\]
and this probability is also equal to $M_e/n$ due to the stationary distribution assigning equal probability to all states in $[n]$.
Therefore the exponent in \eqref{eq:time-for-stationarity} is
\begin{align*}
\frac{1}{n}\sum_{\omega \in [n]^{k+1}} q(M, \tau, \omega)\#(e,\tau,\omega)
&= \frac{1}{n}\sum_{\omega \in [n]^{k+1}} \sum_{j \in [k]} q(M, \tau, \omega,j)\Ind[e = (a(\tau,\omega,j),b(\tau,\omega,j))]
\\&= \frac{1}{n}\sum_{j \in [k]} M_e
\\&= \frac{kM_e}{n}.
\end{align*}

Plugging this into \eqref{eq:time-for-stationarity}, we have
\[ S_{\tau} \ge
Z\prod_{e \in [n] \times [n]} M_{e}^{kM_e/(nZ)}
\ge Z\left(\prod_{e \in [n] \times [n]} M_{e}^{M_e}\right)^{k/(nZ)}.\]
Since $f(x) = x^x$ is log-convex, we have $\prod_{i=1}^{N} x_i^{x_i} \ge \bar{x}^{N\bar{x}}$ for any sequence $x_1, \dots, x_N$, with average value $\bar{x} = \frac{1}{N}\sum_{i=1}^N x_i$.
Applying this to the $dn$ non-zero entries $M_e$ with average value $1/d$, we find
\[ S_{\tau}
\ge Z(\bar{M_e}^{nd\bar{M_e}})^{k/(nZ)} = \frac{Z}{d^{k/Z}}\]
Now plugging this into \eqref{eq:catalan-subset} along with the definition of $S_{\tau}$,
\[ \frac{1}{n}\trace M^{2k} \ge \sum_{\tau \in \cT_k} S_{\tau} \ge \frac{Z|\cT_k|}{d^{k/Z}}. \]
The number $|\cT_k|$ of different Dyck paths of length $2k$ is given by the $k$th Catalan number, equal to $\frac{1}{k+1}\binom{2k}{k}$, which is at least $4^k/((k+1)\sqrt{\pi (k+1/2)}) \ge (1 - O(k^{-5/2}))\cdot 4^k/(k\sqrt{\pi k})$ by Stirling's approximation.
Therefore,
\[ \frac{1}{n}\trace M^{2k} \ge (1-O(k^{-5/2}))\cdot\frac{4^k Z}{k\sqrt{\pi k}\,d^{k/Z}}.\]
By \autoref{lem:Z},
\[ \frac{1}{n}\trace M^{2k} \ge \left(1-\frac{O(1)}{k^2\sqrt{k}} - \frac{2k}{\sqrt{d}}\right)\frac{4^k}{k\sqrt{\pi k}\,d^{k/(1 - \frac{2k}{\sqrt{d}})}}.\]
Assuming now that $k \le \sqrt[4]{d}$, and $(\ln d)^2 \ge 4$ and $d \ge 4^4$,
\begin{align*}
\frac{1}{n}\trace M^{2k}
&\ge \left(1-o_k(1)\right)\frac{4^k}{k\sqrt{\pi k}\,d^{k(1 + \frac{4k}{\sqrt{d}})}}
\\&= \left(1-o_k(1)\right) e^{-\frac{4k^2\ln d}{\sqrt{d}}}\frac{4^k}{k\sqrt{\pi k}\,d^k}
\\&\ge \left(\frac{1-o(1)}{e\sqrt{\pi}\,k\sqrt{k}}\right)\frac{4^k}{d^{k-4}}
\\&\ge \left(e^{-(3\ln k +1+o(1))/(4k)}\cdot e^{-2(\ln d)/k}\cdot\frac{2}{\sqrt{d}}\right)^{2k}
\\&\ge \left(\left(1-\frac{3\ln k + 1 + 8\ln d + o(1)}{4k}\right)\frac{2}{\sqrt{d}}\right)^{2k}
\\&\ge \left(\frac{2 - \frac{8\ln d + 3\ln k + 1 + o(1)}{2k}}{\sqrt{d}}\right)^{2k}
. \qedhere
\end{align*}
\end{proof}


\subsection{The Test Vector}
\label{sec.matrix}

Let us now prove \autoref{lem.almostdoublystoc}.

\begin{proof}[Proof of \autoref{lem.almostdoublystoc}]
We can assume without loss of generality that $M_{a,b} \leq \frac{2}{\sqrt{d}}$ for all $a, b \in [n]$ as otherwise, \autoref{lem.doublystoc-smalledgewt} implies the theorem statement. Hence, we may apply \autoref{lem:almost-doubly-stochastic} and find that any choice of $k \leq d^{1/4}$ admits $\trace M^{2k} \geq n \cdot \Big( \big(2 - O(\frac{\ln k}{k}) \big) / \sqrt{d} \Big)^{2k}$. However, $\trace M^0 = n$ and as
\begin{align*}
\trace M^{2k}
&= \frac{\trace M^{2k}}{\trace M^{2k - 2}}
\cdot \frac{\trace M^{2k - 2}}{\trace M^{2k - 4}}
\cdot \ldots
\cdot \frac{\trace M^{2}}{\trace M^{0}}
\cdot \trace M^0 \\
&= \frac{\trace M^{2k}}{\trace M^{2k - 2}}
\cdot \frac{\trace M^{2k - 2}}{\trace M^{2k - 4}}
\cdot \ldots
\cdot \frac{\trace M^{2}}{\trace M^{0}}
\cdot n \\
&\geq \bigg( \frac{2 - O\left(\frac{\ln k}{k}\right)}{\sqrt{d}} \bigg)^{2k} \cdot n
\end{align*}
there must exist at least one choice of $\ell$ such that
\begin{equation*}
\trace M^{2\ell} \geq \left( \left(2 - O\left(\frac{\ln k}{k}\right) \right) / \sqrt{d} \right)^2 \cdot \trace M^{2 \ell - 2}
\end{equation*}
Furthermore, choosing the smallest such $\ell$ guarantees that $\trace M^{2\ell} \geq \Big( \big( 2 - O(\frac{\ln k}{k}) \big) / \sqrt{d} \Big)^{2k} n$, since then
\[ \trace M^{2\ell} = \frac{\trace M^{2k}}{\prod_{\ell < i \le k} (\trace M^{2i}/\trace M^{2i-2})} \ge \frac{\Big( \big( 2 - O(\frac{\ln k}{k}) \big) / \sqrt{d} \Big)^{2k} n}{ \prod_{\ell < i \le k} 1}.\]
Hence, let $a \geq \big( 2 - O(\frac{\ln k}{k}) \big)/\sqrt{d}$ such that $\trace M^{2\ell} = a^2 \cdot \trace M^{2 \ell - 2}$ and consider the PSD matrices:
\begin{equation*}
X = (M^\ell + a M^{\ell - 1})^2
\qquad\qquad\qquad
Y = (M^\ell - a M^{\ell - 1})^2
\end{equation*}

Let $\epsilon > 0$ be the largest value such that for all $Z \succeq \mathbf{0}$ with $S \subseteq \ker Z$,
\begin{equation*}
|\langle Z, M - D + I \rangle| \leq \epsilon \cdot \trace Z.
\end{equation*}
Let $\Pi^{\perp}$ be the orthogonal projector away from $S$, so that $S \subseteq \ker \Pi^{\perp} X \Pi^{\perp}$ and $S \subseteq \ker \Pi^{\perp} Y \Pi^{\perp}$.
Applying the above equation to both matrices admits the following.
\begin{equation*}
\iprod{\Pi^{\perp} X \Pi^{\perp}, M - D + I}
\leq \epsilon \cdot \trace \big( \Pi^{\perp} X \Pi^{\perp} \big)
\qquad\qquad
-\iprod{\Pi^{\perp} Y \Pi^{\perp}, M - D + I}
\leq \epsilon \cdot \trace \big( \Pi^{\perp} Y \Pi^{\perp} \big)
\end{equation*}
Adding these inequalities together derives
\begin{align*}
\iprod{\Pi^{\perp} (X - Y) \Pi^{\perp}, M - D + I}
\leq \epsilon \cdot \trace \Big( \Pi^{\perp} \big( X + Y \big) \Pi^{\perp} \Big)
\end{align*}
We will then note that
\begin{equation*}
X - Y
= \big( M^\ell + a M^{\ell - 1} \big)^2 - \big( M^\ell - a M^{\ell - 1} \big)^2
= 4a M^{2\ell - 1}
\end{equation*}
\begin{equation*}
X + Y
= \big( M^\ell + a M^{\ell - 1} \big)^2 + \big( M^\ell - a M^{\ell - 1} \big)^2
= 2 (M^{2\ell} + a^2 M^{2\ell - 2})
\end{equation*}
and thus
\begin{equation}
\label{eq:sumdifference}
4a \cdot \iprod{\Pi^{\perp} M^{2\ell - 1} \Pi^{\perp}, M - D + I}
\leq 2\epsilon \cdot \trace \Big( \Pi^{\perp} \big( M^{2\ell} + a^2 M^{2\ell - 2} \big) \Pi^{\perp} \Big).
\end{equation}

Now for $M$ that is $\epsilon$-almost doubly stochastic, we have $\big\lvert \ones^{\T} M \ones - \ones^{\T} I \ones \big\rvert \leq \epsilon n$. Denoting $\lambda$ as the Perron--Frobenius eigenvalue of $M$, we have $1 - \epsilon \leq \lambda \leq 1 + \epsilon$. Since $M$ is $2/\sqrt{d}$-almost doubly stochastic, $\|M\| \le (1+2/\sqrt{d})$. Therefore, letting $r = \dim S$ and $v_1, \dots, v_{r}$ be an orthonormal basis for $S$, then using linearity of and the cyclic property of the trace,
\begin{align*}
\iprod{\Pi^{\perp} M^{2\ell - 1} \Pi^{\perp}, M}
&{}= \iprod{M^{2\ell - 1}, M} - 2\sum_{i\in [r]}\iprod{v_iv_i^{\T} \cdot M^{2\ell - 1}, M} + \sum_{i\in [r]}\iprod{v_iv_i^{\T} \cdot M^{2\ell - 1} \cdot v_iv_i^{\T}, M}
\\&{}= \trace M^{2\ell} - 2\sum_{i\in [r]} v_i^{\T} M^{2\ell} v_i + \sum_{i\in [r]} \big( v_i^{\T} M^{2\ell - 1} v_i \big) \big( v_i^{\T} M v_i \big)
\\&{}\ge \trace M^{2\ell} - 2r \cdot \|M^{2\ell}\| - r \cdot \|M^{2\ell - 1}\|\, \|M\|\
\\&{}\ge \trace M^{2\ell} - 3r \cdot \big( 1+2/\sqrt{d} \big)^{2\ell}.
\end{align*}
Also, using the bound on $\iprod{M^{2\ell - 1}, D-I}$ from the lemma assumptions and the fact that $\|D-I\| \le 2/\sqrt{d}$ since $M$ is $2/\sqrt{d}$-almost doubly stochastic,
\begin{align*}
\iprod{\Pi^{\perp} M^{2\ell - 1} \Pi^{\perp}, D-I}
&{}=\iprod{M^{2\ell - 1}, D-I} - 2\sum_{i\in [r]}\iprod{v_iv_i^{\T} M^{2\ell - 1}, D-I} + \sum_{i\in [r]}\iprod{v_iv_i^{\T} M^{2\ell - 1} v_iv_i^{\T}, D-I}
\\&{}\le \frac{\sqrt{d}}{n^{\gamma}}\trace M^{2\ell} - 2\sum_{i\in [r]} v_i^{\T} \big( M^{2\ell-1} (D-I) \big) v_i + \sum_{i\in [r]} \big( v_i^{\T} (D-I) v_i \big) \big( v_i^{\T} M^{2\ell - 1} v_i \big)
\\&{}\le \frac{\sqrt{d}}{n^{\gamma}}\trace M^{2\ell} + 2r \cdot \|M^{2\ell-1}\|\,\|D-I\| + r \cdot \|D-I\|\,\|M^{2\ell - 1}\|
\\&{}\le \frac{\sqrt{d}}{n^{\gamma}}\trace M^{2\ell} + 6r \cdot \big( 1+2/\sqrt{d} \big)^{2\ell-1}/\sqrt{d}.
\end{align*}
Also, since $M^{2\ell} + a^2 M^{2\ell - 2} \succeq \mathbf{0}$ and $\Pi^{\perp}$ is contractive,
\begin{align*}
\trace \Big( \Pi^{\perp} \big( M^{2\ell} + a^2 M^{2\ell - 2} \big) \Pi^{\perp} \Big)
\le \trace \big( M^{2\ell} + a^2 M^{2\ell - 2} \big).
\end{align*}
Therefore, substituting the above inequalities into \eqref{eq:sumdifference},
\[4a\left(\trace M^{2\ell} - 3(1+2/\sqrt{d})^{2\ell}r - \frac{\sqrt{d}}{n^{\gamma}}\trace{M^{2\ell}} - 6r(1+2/\sqrt{d})^{2\ell-1}/\sqrt{d}\right) \le  2\epsilon \,\trace\big( M^{2\ell} + a^2 M^{2\ell - 2} \big).\]
Simplifying,
\[4a\left(\left(1-\frac{\sqrt{d}}{n^{\gamma}}\right)\trace M^{2\ell} - 3r(1+4\sqrt{d})(1+2/\sqrt{d})^{2\ell-1}\right) \le  2\epsilon \,\trace\big( M^{2\ell} + a^2 M^{2\ell - 2} \big).\]
Recalling we chose $a$ to satisfy $\trace M^{2\ell} = a^2 \cdot \trace M^{2 \ell - 2}$,
\[4a\left(\left(1-\frac{\sqrt{d}}{n^{\gamma}}\right)\trace M^{2\ell} - 3r(1+4\sqrt{d})(1+2/\sqrt{d})^{2\ell-1}\right) \le  4\epsilon \,\trace\big( M^{2\ell} \big),\]
and as we set $\ell$ so that $\trace M^{2\ell} \ge n\Big( \big( 2 - O(\frac{\ln k}{k}) \big) / \sqrt{d} \Big)^{2k}$,
\[4a\left(1 - \frac{\sqrt{d}}{n^{\gamma}} - \frac{3r(1+4/\sqrt{d})(1+2/\sqrt{d})^{2\ell-1}}{n\Big( \big( 2 - O(\frac{\ln k}{k}) \big) / \sqrt{d} \Big)^{2k}}\right) \le  4\epsilon.\]
Finally, $r \leq n^\delta$ therefore we must have
\[
\epsilon
\;\ge\; a\left(1 - \frac{\sqrt{d}}{n^{\gamma}} - O\left(\frac{d^kn^{\delta}}{n}\right)\right)
\;\ge\; \frac{2}{\sqrt{d}} - O \bigg( \frac{\ln k}{k\sqrt{d}} \bigg) - O\bigg( \frac{d^{k}}{n^{1-\delta} \sqrt{d}}\bigg) - O\bigg( \frac{1}{n^{\gamma}}\bigg).
\]
Choosing $k = d^{1/4}$ obtains the result.
\end{proof}


\subsection{Completing the Proof}
\label{sec.final-sep}

Culminating this section, we now complete the proof of \autoref{thm:lower-bound-spectral} restated as follows.

\lowerboundspectral*

\begin{proof}
In what follows, we assume that each vertex $i$ in $H$ has weighted degree bounded by $1 \pm 2/\sqrt{d}$. Otherwise, we use \autoref{lem.doublystoc-smallwtddegree} to conclude the theorem statement.

We begin by considering cases where $G = \bar K_n$. The Laplacian of $\bar K_n$ is given by
\begin{equation*}
L_{\bar K_n} = I - \tfrac{1}{n} \cdot \ones \ones^{\T}.
\end{equation*}
$H$ $\epsilon$\! spectrally sparsifies $\bar{K}_n$ thus requires $L_{H}$ to satisfy for all $X \succeq {\bf 0}$,
\[
    (1-\epsilon) \iprod{X, L_{\bar K_n}}
    \leq \iprod{X, L_{H}}
    \leq (1+ \epsilon) \iprod{X, L_{\bar K_n}},
\]
\[
    (1-\epsilon) \iprod{X, I - \tfrac{1}{n} \cdot \ones \ones^{\T}}
    \leq \iprod{X, L_{H}}
    \leq (1+ \epsilon) \iprod{X, I - \tfrac{1}{n} \cdot \ones \ones^{\T}}.
\]
If $\ones \in \ker X$, then
\[
    (1-\epsilon) \iprod{X, I}
    \leq \iprod{X, L_{H}}
    \leq (1+ \epsilon) \iprod{X, I},
\]
\[
    -\epsilon \trace X
    \leq \iprod{X, I - L_{H}}
    \leq \epsilon \trace X,
\]
\[
    -\epsilon \trace X
    \leq \iprod{X, A_{H} - D_{H} + I}
    \leq \epsilon \trace X.
\]
For the first two conditions of \autoref{thm:lower-bound-spectral}, our strategy will be to use \autoref{lem.almostdoublystoc} to demonstrate the existence of $X$ such that $S_{1} = \operatorname{span} ( \{ \ones \} ) \subseteq \ker X$ and $\epsilon \geq \frac{\iprod{X, A_H - D_H + I}}{\trace X}$ is large.

Suppose $H$ satisfies condition (1). Observe that the subspace of $\R^V$ given by $S_1$ admits $\lVert v \rVert_{\infty} / \lVert v \rVert_2 \leq 1/\sqrt{n}$ for all $v \in S_1$. Since $A_H$ is strictly doubly stochastic, $D_H = I$ and thus for every $\gamma > 0$, and every $k > 0$
\begin{equation*}
\big\lvert \iprod{A_H^{2k - 1}, D_H - I} \big\rvert
= \big\lvert \iprod{A_H^{2k - 1}, \mathbf{0}} \big\rvert
= 0
< \frac{\sqrt{d}}{n^\gamma}
\end{equation*}
As $A_H$ has at most $dn$ non-zero entries, applying \autoref{lem.almostdoublystoc} using $S_1$, $\delta = 1/2$, and $M = A_H$, then dropping the $O(\frac{1}{n^{\gamma}})$ yields $X \succeq \mathbf{0}$ satisfying the required
\begin{equation*}
\frac{2}{\sqrt d} - O \bigg( \frac{(\ln d)^3}{d} \bigg) - O\bigg( \frac{d^{\sqrt{d} / (\ln d)^2}}{n^{1-\delta} \cdot \sqrt{d}}\bigg) - O\left(\frac{1}{n^{1/2}}\right)
\leq \frac{\iprod{ X, A_H - D_H + I }}{\trace X}
\leq \epsilon.
\end{equation*}

Suppose $H$ satisfies condition (2). Since the weighted degree of any vertex in $H$ is at most $1 \pm \frac{2}{\sqrt{d}}$, $A_H$ is $2/\sqrt{d}$-almost doubly stochastic. The entries of $D_H - I$ must then be bounded by $2/\sqrt{d}$ in magnitude. If at most $B$ vertices of $H$ participate in odd cycles of length at most $2g$, then
\begin{equation*}
\big\lvert \iprod{ A_H^{2k - 1}, D_H - I} \big\rvert
\leq \Big\lvert \iprod{ A_H^{2k - 1}, \tfrac{2}{\sqrt{d}} \cdot I} \Big\rvert
\leq \tfrac{2 \cdot B}{\sqrt{d}}
\end{equation*}
for all $k \leq g$ since $\big\lvert \big\{ i \in V : \big( A_H^{2k-1} \big)_{i,i} > 0 \big\} \big\rvert < B$ in addition to $\big( A_H^{2k-1} \big)_{i,i} < 1$ for all $i$. The latter fact follows by the $(i,i)$-th entry of $A_H^{2k-1}$ representing the probability that a random walk starting from $i$ ends at $i$ after $2k - 1$ steps. We assume for now that each edge of $H$ has weight at most $\frac{2}{\sqrt{d}}$, in which case $\big( A_H \big)_{a,b} \leq 2/\sqrt{d}$ for all $a, b \in [n]$. $A_H$ being $2/\sqrt{d}$-almost doubly stochastic then implies for all $k \leq d^{1/4}$
\begin{equation*}
\trace M^{2k}
\geq n \bigg( \frac{2 - O(\frac{\ln k}{k})}{\sqrt{d}} \bigg)^{2k}
\end{equation*}
by \autoref{lem:almost-doubly-stochastic}. Selecting $g \leq d^{1/4}$ and $B = o(n)$, we have for all $k \leq d^{1/4}$:
\begin{equation*}
\bigg\lvert \frac{\iprod{ A_H^{2k - 1}, D_H - I}}{\trace A_H^{2k}} \bigg\rvert
\leq \frac{2 \cdot B / \sqrt{d}}{n \cdot \big( 2 / \sqrt{d} - O \big( \frac{\ln k}{k\sqrt{d}} \big) \big)^{2k}}
\leq \frac{2 / \sqrt{d}}{ \omega_n(1) \cdot \big( 2 / \sqrt{d} - O \big( \frac{\ln k}{k\sqrt{d}} \big) \big)^{2k}}
\leq o_{n,d}(1).
\end{equation*}
For large enough $n$, there then exists $\gamma > 0$ such that $\frac{\big\lvert \iprod{ A_H^{2 k - 1}, D_H - I} \big\rvert}{\trace A_H^{2k}} \leq \frac{\sqrt{d}}{n^{\gamma}}$ for all $k \leq d^{1/4}$. With the same choice of $S_1$, $\delta$, and $M$ from case (1), \autoref{lem.almostdoublystoc} implies the required lower bound on $\epsilon$.

Revisiting our assumption that $( A_H )_{a,b} \leq 2/\sqrt{d}$ for all $a, b \in [n]$, if there is an edge $(a, b)$ of $H$ such that $\big( A_H \big)_{a,b} > 2/\sqrt{d}$, then using \autoref{lem.doublystoc-smalledgewt} with same choice of $S_1$, $\delta$, and $M$ recovers the required bound on $\epsilon$.

Finally, we consider case (3) where $H$ is bipartite and $G = \bar{K}_{n/2, n/2}$. Let $v \in \R^V$ be the balanced vector indicating the partitions of $\bar{K}_{n/2, n/2}$; indices corresponding to the $n/2$ vertices on one side of the partition are $1$ and the remaining $n/2$ are $-1$. We have that $L_{\bar{K}_{n/2, n/2}} = I - \frac{1}{n} \cdot \ones \ones^{\T} + \frac{1}{n} \cdot v v^{\T}$. The condition that $H$ $\epsilon$\! spectrally sparsifies $\bar{K}_{n/2, n/2}$ then requires $L_H$ satisfy for all $X \succeq \mathbf{0}$
\[
    (1-\epsilon) \iprod{X, I - \tfrac{1}{n} \cdot \ones \ones^{\T} + \tfrac{1}{n} \cdot v v^{\T}}
    \leq \iprod{X, L_{H}}
    \leq (1+ \epsilon) \iprod{X, I - \tfrac{1}{n} \cdot \ones \ones^{\T} + \tfrac{1}{n} \cdot v v^{\T}}.
\]
If $\{ \ones, v \} \subseteq \ker X$, then
\[
    -\epsilon \trace X
    \leq \iprod{X, A_{H} - D_{H} + I}
    \leq \epsilon \trace X.
\]
Hence, we will use \autoref{lem.almostdoublystoc} to demonstrate the existence of $X \succeq \mathbf{0}$ such that $S_2 = \operatorname{span}(\{ \ones, v \}) \subseteq \ker X$ and $\epsilon \geq \frac{\iprod{X, A_H - D_H + I}}{\trace X}$ is large.

Observe that we have $\lVert v \rVert_{\infty} / \lVert v \rVert_2 \leq 2/\sqrt{n}$ for all $v \in S_2$. We also have that
\begin{equation*}
\big\lvert \iprod{A_H^{2k - 1}, D_H - I} \big\rvert
= 0
\end{equation*}
for all $k \leq d^{1/4}$ since all diagonal entries of $A_H^{2\ell - 1}$ are zero: $H$ is bipartite and thus cannot possess any cycles with odd length. For any choice of $\gamma > 0$, we have for all $k \leq d^{1/4}$
\begin{equation*}
\frac{\big\lvert \iprod{A_H^{2k - 1}, D_H - I} \big\rvert}{\trace A_H}
\leq O \bigg( \frac{\sqrt{d}}{n^{\gamma}} \bigg)
\end{equation*}
The weighted degree of any vertex in $H$ is at most $1 \pm \frac{2}{\sqrt{d}}$ and hence $A_H$ is $2/\sqrt{d}$-almost doubly stochastic. Applying \autoref{lem.almostdoublystoc} using $S_2$, $\delta = 1/2$, and $M = A_H$ yields $X \succeq \mathbf{0}$ such that
\begin{equation*}
2/\sqrt{d} - O \bigg( \frac{(\ln d)^3}{d} \bigg) - O\bigg( \frac{d^{\sqrt{d} / (\ln d)^2}}{n^{1/2} \cdot \sqrt{d}}\bigg)
\leq \frac{\iprod{ X, A_H - D_H + I }}{\trace X}
\leq \epsilon
\end{equation*}
as required.
\end{proof}

\section{Separation Between Cut and Spectral Sparsification}
\label{sec:separation}

We will now show a separation between cut and spectral sparsification of random $\log n$-regular graphs. To begin, we show prove that the complete graph can be cut sparsified past the Ramanujan bound.

\main*

\begin{proof}
Let $G$ be the complete graph on $n$ vertices and $\epsilon \leq \Big( 2 \sqrt{\frac{2}{\pi}} + o_{n,d}(1) \Big) / d$ and consider a random regular graph $H = (V,E)$ drawn from $\gReg{n}{d}$ for a large enough constant $d$ with edges weighted by $\frac{n-1}{d}$. By Theorem~\ref{thm:large-cuts}, there is at least a $1-e^{-\Omega(n/\log n)}$ chance $H$ $\epsilon$ cut sparsifies $G$ when restricting to cuts $S \subset V$ satisfying $0.01 n \le \lvert S \rvert \le \frac{n}{2}$.

Concurrently, Theorem~\ref{thm:small-cuts-crossing} states that for any fixed $k \leq 0.01n$, with probability at least $1-\binom{n}{k}^{-1.01}$, $H$ $\epsilon$ cut sparsifies $G$ on all cuts $S \subset V$ of the same size $|S| = k$. Performing a union bound over all sizes below, we find that with probability at least $1 - \sum_{k=1}^{0.01n} \binom{n}{k}^{-0.01}$, $H$ $\epsilon$ cut sparsifies $G$ on all $S \subset V$ with $|S| \le 0.01 n$. This probability is at least $1-O(n^{-0.01})$. By a final union bound over the cases $\lvert S \rvert \leq 0.01n$ and $0.01n \leq \lvert S \rvert \leq \frac{n}{2}$, $H$ is an $\epsilon$ cut sparsifier of $G$ with probability at least $1 - o_n(1)$ as required.
\end{proof}

We then show that a random $\log n$ regular graph satisfies the pseudo-girth conditions required by Theorem~\ref{th.speclower}.

\begin{theorem}\label{th.pseudogirth}
If $G$ is a random regular graph drawn from $\gReg{n}{\log n}$, and $g$ is a fixed constant then the following occur.
\begin{enumerate}[1.]
\item With probability $1$, for every vertex $v$ of $G$, the number of vertices of $G$ reachable from $v$ via paths of length at most $g$ is $O\big( (\log n)^g \big)$

\item Let $B$ be the set of vertices v such that v participates in a cycle contained in the vertex-induced subgraph of G, induced by vertices of distance at most $2g$ from $v$. Then $\lvert B \rvert \leq O \big( (\log n)^{4g+1} \big)$ with probability $1 - o_n(1)$ over the choice of $G$.
\end{enumerate}
\end{theorem}
\begin{proof}
The first property immediately follows from the fact that the combinatorial degree is at most $\log n$. For the second part, fix a vertex $v$ and consider the probability, over the choice of $G$, that $v \in B$. By applying the principle of deferred decision, we first generate the $\log n$ neighbors of $v$, then the additional neighbors of those neighbors, and so on. Every time we make a decision about how to match a particular vertex $u$ in one of the $\log n$ matchings, the probability of hitting a previously seen vertex is at most $O\big( (\log n)^{2g}/n \big)$ and so the probability that we create a cycle is at most $O\big( (\log n)^{4g}/n \big)$. The conclusion of the theorem follows by applying Markov's inequality.
\end{proof}

We conclude by proving the separation between cut and spectral sparsification stated by Theorem \ref{th.separation}.

\separation*

\begin{proof}
Fix $d$. If $G$ is a random $\log n$-regular graph drawn from $\gReg{n}{\log n}$, then, for every fixed $d$ there is a $1-o_n(1)$ probability that there are $o_n(n)$ vertices that see a cycle within distance $d^{1/4}$ and there are $o_n(n)$ vertices in the ball of radius $d^{1/4}$ around each vertex. Note that the above properties will also hold for any edge-subgraph $H$ of $G$.

From Theorem \ref{th.speclower} we have that, with $1-o_n(1)$ probability over the choice of $G$, if a weighted edge-induced subgraph $H$ of $G$ of average degree $d$ is an $\epsilon$ spectral sparsifier of the clique, then $\epsilon \geq (2 - O(d^{-3/4}) - o_n(1))/\sqrt d$. From
\cite{bordenave2019new} we have that with $1-o_n(1)$ probability the graph $G$ is an $O(1/\sqrt {\log n})$ spectral sparsifier (and also cut sparsifier) of the clique, and so if a weighted edge-induced subgraph $H$ of $G$ of average degree $d$ is an $\epsilon$ cut sparsifier of the $G$, then again $\epsilon \geq (2 - O(d^{-3/4}) - o_n(1))/\sqrt d$

Since we constructed $G$ as the union of $\log n$ random matchings, $G$ contains, for large enough $n$, a random $d$-regular graph from $\gReg{n}{d}$ as an edge-induced subgraph (for example, consider the first $d$ of the $\log n$ matchings used to construct $G$). We can deduce from Theorem \ref{th.main} that, with $1-o_n(1)$ probability, $G$ contains as a weighted edge-induced subgraph a graph $H$ that has average degree $d$ and is a $(1.595\ldots + o_{n,d}(1))/\sqrt d$ cut sparsifier of the clique.

We conclude that with $1-o_n(1)$ probability over the choice of $G$, there is a weighted edge-induced subgraph $H$ of $G$ such that $H$ has average degree $d$ and is a $(1.595\ldots + o_{n,d}(1))/\sqrt d$ cut sparsifier of $G$
\end{proof}


\section{Acknowledgments}
The work of JS and LT on this project has received funding from the European Research Council (ERC) under the European Union's Horizon 2020 research and innovation programme (grant agreement No. 834861). AC is supported by NSF DGE 1746045. Most of this work was done while AC was a visiting student at Bocconi University. The authors would like to thank Andrea Montanari for pointing us to \cite{jagannath2017unbalanced}, and the Physics and Machine Learning group at Bocconi, particularly Enrico Malatesta, for patiently explaining the Parisi equations and the replica method to us.

\bibliographystyle{alpha}
\bibliography{main.bib}


\appendix 
\newpage

\section{Expected free energy}

In Section~\ref{sec.linear}, we needed the following statement in the proof of Theorem~\ref{thm:psi-to-random-graphs}.
\begin{lemma}
\label{lem:js17-fix}
For $\eps_n \to 0$ slowly enough as $n \to \infty$, it holds that for all $T$,
\[ \liminf_{n\to \infty} \E_W\;\max_{\sigma \in A_n(T(\alpha), \eps_n)}\, \frac{1}{n}H_W^{(0)}(\sigma) \le \inf_{\nu,\lambda} \cP^1_{T(\alpha)}(\nu, \lambda). \]
\end{lemma}

\begin{proof}
In the proof of \cite[Theorem 1.2]{jagannath2017unbalanced}, it is stated that
\[ \frac{1}{\beta}F_N(\beta, \xi; A_N) - \frac{\log |\Sigma|}{\beta} \le GS_N(A_N) \le \frac{1}{\beta}F_N(\beta, \xi; A_N) + \frac{\log |\Sigma|}{\beta}. \]
By linearity of expectation,
\[ \E_J GS_N(A_N) \le \E_J \frac{1}{\beta}F_N(\beta, \xi; A_N) + \frac{\log |\Sigma|}{\beta}. \]
Taking the limit on both sides,
\[ \liminf_{N \to \infty} \E_J GS_N(A_N) \le \liminf_{N \to \infty} \E_J \frac{1}{\beta}F_N(\beta, \xi; A_N) + \frac{\log |\Sigma|}{\beta}. \]
Since $\E_J \frac{1}{\beta}F_N(\beta, \xi; A_N(T,\eps_N)) \to \frac{1}{\beta} F(\beta, \xi; T)$ as $N \to \infty$ by \cite[Equation 1.12]{jagannath2017unbalanced}, we substitute
\[ \liminf_{N \to \infty}  \E_J GS_N(A_N) \le \frac{1}{\beta}F(\beta, \xi; T) + \frac{\log |\Sigma|}{\beta}. \]
so, taking $\beta \to \infty$,
\[\liminf_{N \to \infty} \E_J GS_N(A_N) \le E(\xi; T). \]

That last inequality completes the proof, together with the facts that, translating from the notation of Theorem~\ref{thm:psi-to-random-graphs} to that of \cite[Theorem 1.2]{jagannath2017unbalanced},
\[\E_W\;\max_{\sigma \in A_n(T(\alpha), \eps_n)}\, \frac{1}{n}H_W^{(0)}(\sigma) = \E_J \sup_T GS_N(A_N) \] 
and, as defined in the proof of \cite[Theorem 1.2]{jagannath2017unbalanced},
\[E(\xi; T) = \inf_{\nu,\lambda} \cP^1_{T(\alpha)}(\nu, \lambda). \qedhere \]
\end{proof}

\section{Analytic Inequalities}
\label{sec.analytic}

In this section, we prove Lemmas~\ref{lem:log-taylor} and~\ref{lem:log-calc}. These are inequalities used to bound the exponent of the tail probability in Lemma~\ref{lem:concentration-generic} (subsequently deriving Lemma~\ref{lem:concentration-cases}) under two cases: when $k \geq \Omega(n / \sqrt{d})$ and when $k < O(n / \sqrt{d})$. We first restate, and prove Lemma~\ref{lem:log-taylor}.

\logtaylor*

\begin{proof}
Proceed by expanding $\delta \ln \big( \frac{\delta}{C} + 1 \big)$ via its Taylor approximation 
\begin{equation*}
\delta \ln \bigg( \frac{\delta}{C} + 1 \bigg)
= \delta \cdot \bigg( \sum_{t=1}^\infty (-1)^{t+1} \frac{\delta^t}{C^t t} \bigg)
= \sum_{t=1}^\infty (-1)^{t+1} \frac{\delta^{t+1}}{C^t t}
\end{equation*}

similarly for $C \ln \big( \frac{\delta}{c} + 1 \big)$, we have 
\begin{equation*}
C \cdot \ln \bigg( \frac{\delta}{C} + 1 \bigg)
= C \cdot \bigg( \sum_{t=1}^\infty (-1)^{t+1} \frac{\delta^t}{C^t t} \bigg)
= \sum_{t=1}^\infty (-1)^{t+1} \frac{\delta^t}{C^{t-1} t}
= \delta + \sum_{t=1}^\infty (-1)^t \frac{\delta^{t+1}}{C^t (t+1)}
\end{equation*}

Combining the two expansions, we derive 
\begin{equation*}
\big( \delta + C \big) \ln \bigg( \frac{\delta}{C} + 1 \bigg)
= \delta + \sum_{t=1}^\infty (-1)^{t+1} \bigg( \frac{\delta^{t+1}}{C^t} \bigg) \bigg( \frac{1}{t} - \frac{1}{t+1} \bigg)
\geq \delta + \frac{\delta^2}{3C} \qedhere
\end{equation*}
\end{proof}

The following is a proof of Lemma~\ref{lem:log-calc}.

\logcalc*

\begin{proof}
Denote $f(C, \delta) = \big( \delta + C \big) \ln \big( \frac{\delta}{C} + 1 \big) - \delta - \frac{1}{2C} \cdot \delta \ln \delta$. It suffices to demonstrate $f(C, \delta) \geq 0$ for all $\delta \geq C \geq 1$. To see this, first note $f(C, \delta) \geq 0$ for all $\delta = C \geq 1$ as we have 
\begin{equation*}
f(C, \delta) 
= \big( \delta + C \big) \ln \bigg( \frac{\delta}{C} + 1 \bigg) - \delta - \frac{1}{2C} \cdot \delta \ln \delta
= 2 \delta \cdot \ln(2\delta) - \delta - \frac{\ln \delta}{2} 
\end{equation*}

which is true for any $\delta \geq 1$. We next compute $\frac{\partial f}{\partial \delta}$ as follows.
\begin{equation*}
\frac{\partial f}{\partial \delta} 
= \ln \bigg( \frac{\delta}{C} + 1 \bigg) - \frac{1}{2C} \cdot \big( \ln \delta + 1 \big)
= \ln \bigg( \frac{\delta / C + 1}{\delta^{1/2C}} \bigg) - \frac{1}{2C}
\end{equation*}

If we can show that $\frac{\partial f}{\partial \delta} \geq 0$ for all $\delta \geq C \geq 1$, then we would have that $f$ is non-negative along $\delta = C$, and non-decreasing along the positive $\delta$ direction past the $\delta = C$ line. It must then be that $f$ is non-negative for all $\delta \geq C \geq 1$. Towards this, observe it is equivalent to demonstrate
\begin{equation*}
\bigg( \frac{\delta}{C} + 1 \bigg)^{2C} \geq \delta e
\end{equation*}

With $g(C, \delta) = \big( \frac{\delta}{C} + 1 \big)^{2c} - \delta e$, we notice that for all $\delta = C \geq 1$ we have 
\begin{equation*}
g(C, \delta) 
= \bigg( \frac{\delta}{C} + 1 \bigg)^{2C} - \delta C
= 2^{2\delta} - \delta e
\geq 1
\end{equation*}

with the second equality holding via $\frac{\delta}{C} \geq 1$. Meanwhile, observe that 
\begin{equation*}
\frac{\partial g}{\partial \delta} 
= 2 \bigg( \frac{\delta}{C} + 1 \bigg)^{2C-1} - e
\geq 2 \cdot 2^{2-1} - e
\geq 0
\end{equation*}

Consequently, $g(C, \delta) \geq 0$ and so $\big( \frac{\delta}{C} + 1 \big)^{2C} \geq \delta e$ implying $\frac{\partial f}{\partial \delta} \geq 0$ as required.
\end{proof}

\section{Almost Doubly Stochastic Matrices}
\label{sec.almost-doubly-stoc}

We compute a lower bound on the trace for almost doubly stochastic matrices. 
This will follow the proof of \autoref{lem:doubly-stochastic}, differing only in a few parameters.
We will again use the same notation for walks introduced in \autoref{sec:walks}, but now instead of probabilities of transitions in a Markov chain, we simply have weights of weighted walks.

\begin{lemma}
\label{lem:almost-Z}
Let $M \in \R^{n \times n}$ be a symmetric $\frac{2}{\sqrt{d}}$-almost doubly stochastic matrix such that
$M_{a,b} \le \frac{2}{\sqrt{d}}$ for all $a,b \in [n]$.
For $\tau \in \cT_k$, let
\[ Z = \frac{1}{n}\sum_{\omega \in \Omega_{\tau}} q(M, \tau, \omega).\]
Then 
\[ Z \ge 1 - \frac{4k}{\sqrt{d}}. \]
\end{lemma}
\begin{proof}
Let $\Omega_{\tau,j} \subseteq [n]^{j+1}$ and  $q(M,\tau,\omega,j) = \prod_{i \in [j]} M_{a(\tau,\omega,i),b(\tau,\omega,i)}$ be as in \autoref{lem:Z}.
We set up an induction with the inductive hypothesis that 
\[ \frac{1}{n}\sum_{\omega \in \Omega_{\tau,j+1}} q(M, \tau, \omega, j+1) \ge 1-\frac{4j}{\sqrt{d}}. \]

In the base case where $j = 0$, there are $n$ elements of $\Omega_{\tau,1}$, and $q(M,\tau,\omega,0)$ is an empty product, so the sum is $n$.

In the inductive step, we again partition $\Omega_{\tau,j+1}$ into prefix sets $S_{\omega}$ for $\omega \in \Omega_{\tau,j}$ and find the  $z_{\omega} \in [n]$ so that $\omega \circ (z_{\omega}) \not\in S_{\omega}$. 
So
\[
\sum_{s \in S_{\omega}} q(M, \tau, s, j+1) = \sum_{x \ne z_{\omega}} q(M,\tau,\omega,j) M_{a(\tau,\omega,j),x}
\]
Since $M$ is $2/\sqrt{d}$-almost doubly stochastic and the entries of $M$ are at most $2/\sqrt{d}$,
\[
\sum_{s \in S_{\omega}} q(M, \tau, s, j+1) = (1 - M_{a(\tau,\omega,j),z})\, q(M,\tau,\omega,j) \ge \left(1-\frac{4}{\sqrt{d}}\right) q(M,\tau,\omega,j)
\]
So since $\{S_{\omega}\}$ is a partition of $\Omega_{\tau,j+1}$, 
\begin{align*}
\frac{1}{n}\sum_{\omega \in \Omega_{\tau,j+1}} q(M, \tau, \omega,j+1)
&= \frac{1}{n}\sum_{\omega \in \Omega_{\tau,j}} \sum_{s \in S_{\omega}} q(M, \tau, s, j+1)
\\&\ge \left(1-\frac{4}{\sqrt{d}}\right)\frac{1}{n}\sum_{\omega \in \Omega_{\tau,j}} q(M, \tau, \omega, j)
\\&\ge \left(1-\frac{4}{\sqrt{d}}\right)\left(1-\frac{4(j-1)}{\sqrt{d}}\right)
\\&\ge 1 - \frac{4j}{\sqrt{d}}. \qedhere
\end{align*}
\end{proof}

Using this, we complete the trace lower bound proof for almost doubly stochastic matrices.

\begin{proof}[Proof of \autoref{lem:almost-doubly-stochastic}]
We again define $S_{\tau}$ and derive the inequality
\begin{equation}
\label{eq:almost-catalan-subset}
    \frac{1}{n}\trace M^{2k} \ge \frac{1}{n}\sum_{\tau \in \cT_k} \sum_{\omega \in \Omega_{\tau}} p(M, W(\tau, \omega))
    = \sum_{\tau \in \cT_k} S_{\tau} 
\end{equation}
as we did in the proof of \autoref{lem:doubly-stochastic}.
We again consider each fixed $\tau$ and let $Z = \tfrac{1}{n}\sum_{\omega \in \Omega_{\tau}} q(M, \tau, \omega)$, and $\#(e,\tau,\omega)$ be the number of different $i$ for which $e = (a(\tau,\omega,i),b(\tau,\omega,i))$, to derive the inequality 
\begin{equation}
\label{eq:almost-time-for-stationarity}
S_{\tau} \ge
Z\prod_{e \in [n] \times [n]} M_{e}^{\sum_{\omega \in [n]^{k+1}} q(M, \tau, \omega)\#(e,\tau,\omega)/(nZ)}.
\end{equation}
We again let
\[ q(M,\tau,\omega,j) = \prod_{i \in [j]} M_{a(\tau,\omega,i),b(\tau,\omega,i)}\]
so that
\[ q(M, \tau, \omega)\#(e,\tau,\omega) = \sum_{j \in [k]} q(M,\tau,\omega,j)\Ind[e = (a(\tau,\omega,j),b(\tau,\omega,j))].\]
Now, instead of a probability, $q(M,\tau,\omega,j)$ is the sum of the weights of ``restarting" walks, with the weight of the sequence of the first $j$ elements of $\omega$ is the product of the edge weights incurred by starting at $\omega_0 \in [n]$ and then, at the $i$th step, either
\begin{itemize}
    \item advances to $\omega_i$, multiplying the weight by $M_{\omega_{i-1},\omega_i}$, if $i$ is an increasing step of $\tau$, or
    \item if instead $i$ is a decreasing step of $\tau$, retreats to some previous step $i' < i-1$ and incurs a weight factor of $M_{\omega_{i'},\omega_i}$ instead.
\end{itemize}
In the vector view that we get by summing over all possible sequences, we start at the all-ones vector $v_0 = \ones$ and at step $i$ either
\begin{itemize}
    \item multiply by $M$ to get $v_i = M_{i-1}$, if $i$ is an increasing step of $\tau$, or
    \item retreats to some previous step $i' < i-1$ so that if $i' - (i-1) = c$ then our vector $v_i = \ones^{\T}M^av_{i'}Mv_{i'}$.
\end{itemize}
Therefore, inductively, the sum
\[\sum_{\omega \in [n]^{k+1}} q(M,\tau,\omega,j)\Ind[e = (a(\tau,\omega,j),b(\tau,\omega,j))]\]
is equal to some product $M_{a,b}e_a^{\T}M^{c_0}\ones\prod_{i} \ones^{\T}M^{c_i}\ones$ where $\sum c_i = j$.
Since $M$ is $\frac{2}{\sqrt{d}}$-almost doubly stochastic, we have $|\ones^{\T}M\ones - \ones^{T}\ones| \le \frac{2n}{\sqrt{d}}$,
Therefore, inductively, $|\ones^{\T}M^c\ones - \ones^{T}\ones| \le \frac{2cn}{\sqrt{d}}$ and
\[ |e_a^{\T}M^{c_0}\ones\prod_{i} \ones^{\T}M^{c_i}\ones - 1| \le \frac{2j}{\sqrt{d}} \]
so that the exponent in \eqref{eq:almost-time-for-stationarity} is 
\begin{align*}
\frac{1}{n}\sum_{\omega \in [n]^{k+1}} q(M, \tau, \omega)\#(e,\tau,\omega)
&= \frac{1}{n}\sum_{j \in [k]} M_{e}\ones^{\T}M^{c_0}\ones\prod_{i} \ones^{\T}M^{c_i}\ones
\end{align*}
So that
\[\frac{1}{n}\sum_{\omega \in [n]^{k+1}} q(M, \tau, \omega)\#(e,\tau,\omega) =  \gamma kMe\]
for $|\gamma - 1| \le 2k/\sqrt{d}$, noting that $\gamma$ is only a function of $\tau$ and not of the edge $e$.

Plugging this into \eqref{eq:almost-time-for-stationarity}, we have
\[ S_{\tau} \ge
Z\prod_{e \in [n] \times [n]} M_{e}^{k\gamma M_e/(nZ)}
\ge Z\left(\prod_{e \in [n] \times [n]} M_{e}^{M_e}\right)^{\gamma k/(nZ)}.\]
Since $f(x) = x^x$ is log-convex, we have $\prod_{i=1}^{N} x_i^{x_i} \ge \bar{x}^{N\bar{x}}$ for any sequence $x_1, \dots, x_N$, with average value $\bar{x} = \frac{1}{N}\sum_{i=1}^N x_i$.
Applying this to the $dn$ non-zero entries $M_e$ with average value $1/d$, we find 
\[ S_{\tau}
\ge Z(\bar{M_e}^{nd\bar{M_e}})^{\gamma k/(nZ)} = \left(\frac{Z}{d^{k/Z}}\right)^{\gamma}. \]
Now plugging this into \eqref{eq:almost-catalan-subset} along with the definition of $S_{\tau}$,
\[ \frac{1}{n}\trace M^{2k} \ge \sum_{\tau \in \cT_k} S_{\tau} \ge |\cT_k|\left(\frac{Z}{d^{k/Z}}\right)^{\gamma}. \]
The number $|\cT_k|$ of different Dyck paths of length $2k$ is given by the $k$th Catalan number, equal to $\frac{1}{k+1}\binom{2k}{k}$, which is asymptotically at least $ (1 - O(k^{-5/2}))\cdot 4^k/(k\sqrt{\pi k})$ by Stirling's approximation.
Therefore,
\[ \frac{1}{n}\trace M^{2k} \ge (1-O(k^{-5/2}))\frac{4^k}{k\sqrt{\pi k}}\,\left(\frac{Z}{d^{k/Z}}\right)^{\gamma}.\]
By \autoref{lem:almost-Z} and the fact that $1 - 2k/\sqrt{d} \le \gamma \le 1 + 2k/\sqrt{d}$,
\[ \frac{1}{n}\trace M^{2k} \ge (1-O(k^{-5/2}))\left(1 - \frac{2k}{\sqrt{d}}\right)^{(1 + 2k/\sqrt{d})}\frac{4^k}{k\sqrt{\pi k}\,d^{k(1 + 2k/\sqrt{d})/(1 - 2k/\sqrt{d})}}.\]
Assuming now that $k \le d^{1/4}$ and $(\ln d)^2 \ge 4$ and $d \ge 4^4$,
\begin{align*}
\frac{1}{n}\trace M^{2k}
&\ge (1-o_k(1))\frac{4^k}{k\sqrt{\pi k}\,d^{k(1 + 8k/\sqrt{d})}}
\\&= (1-o_k(1))e^{-(\frac{8k^2\ln d}{\sqrt{d}})}\frac{4^k}{k\sqrt{\pi k}\,d^k}
\\&\ge \left(\frac{1-o_k(1)}{\sqrt{\pi}k\sqrt{k}}\right)\frac{4^k}{d^{k-8}}
\\&\ge \left((e^{-(16\ln d + 3\ln k + 1 + o(1))/(4k)})\frac{2}{\sqrt{d}}\right)^{2k}
\\&\ge \left(\frac{2 - \frac{16\ln d + 3\ln k + 3 + o(1)}{2k}}{\sqrt{d}}\right)^{2k}. \qedhere
\end{align*}
\end{proof}

\section{Assumptions Regarding the Sparsifier}
\label{sec.sparsifier-assumptions}

In this section, we prove the lemmas allowing us to make assumptions on the structure of the sparsifier $H$. We begin by proving \autoref{lem.doublystoc-smallwtddegree} which states one can assume $H$ has weighted degree bounded between $1 \pm 2/\sqrt{d}$, otherwise the sparsifier's error is lower bounded appropriately.

\smallwtddegree*

\begin{proof}
Let us suppose that there is $i \in V$ such that $\big( D_H \big)_{i,i} > 1 + 2/\sqrt{d}$. We then have
\begin{equation*}
e_i^{\T} L_H e_i > 1 + \frac{2}{\sqrt{d}}
\end{equation*}

If $H$ is an $\epsilon$\! spectral sparsifier of $G$ then for any $x \in \R^V$,
\begin{equation*}
x^{\T} L_H x \leq (1 + \epsilon) \; x^{\T} L_G x.
\end{equation*}
When $G = \bar K_n$, we have $L_{\bar K_n} = I - \frac{1}{n} \cdot \ones \ones^{\T}$. Choosing our test vector $x$ to be $e_i$, and applying the condition that $H$ sparsifies $\bar K_n$ using the test vector $e_i$, we determine that
\begin{equation*}
1 + \frac{2}{\sqrt{d}}
< e_i^{\T} L_H e_i
\leq (1 + \epsilon) \; e_i^{\T} L_{\bar K_n} e_i
\leq (1 + \epsilon) \; e_i^{\T} \bigg( I - \frac{1}{n} \cdot \ones \ones^{\T} \bigg) e_i
= (1 + \epsilon) - \frac{1 + \epsilon}{n},
\end{equation*}
or equivalently,
\begin{equation*}
\frac{2}{\sqrt{d}} \cdot \frac{n}{n-1} + \frac{1}{n}
= \frac{2}{\sqrt{d}} \cdot \bigg( 1 + \frac{1}{n-1} \bigg). + \frac{1}{n} 
\leq \epsilon
\end{equation*}
The LHS is at least $2/\sqrt{d}$ for any $n > 0$, thus $\epsilon \geq \frac{2}{\sqrt{d}}$ as required. When $G = \bar{K}_{n/2, n/2}$, its Laplacian matrix is given by $L_{\bar K_{n/2, n/2}} = I - \frac{1}{n} \ones \ones^{T} + \frac{1}{n} v v^{\T}$ where $v$ is the balanced vector indicating the partition. If $H$ sparsifies $\bar K_{n/2, n/2}$, we have that
\begin{equation*}
1 + \frac{2}{\sqrt{d}}
< e_i^{\T} L_H e_i
\leq (1 + \epsilon) \; e_i^{\T} \bigg( I - \frac{1}{n} \cdot \ones \ones^{\T} + \frac{1}{n} \cdot v v^T \bigg) e_i
\leq (1 + \epsilon) - \frac{1+\epsilon}{n} + \frac{1+\epsilon}{n}
\end{equation*}
or equivalently $\epsilon > 2/\sqrt{d}$ as required. The analysis for when $\big( D_H \big)_{i,i} < 1 - 2/\sqrt{d}$ is symmetric; choose $x = e_i$ and use the condition that $H$ sparsifies $G$ only if $x^{\T} L_H x \geq (1 - \epsilon) \; x^{\T} L_G x$ for any $x \in \R^V$.
\end{proof}

We next prove \autoref{lem.doublystoc-smalledgewt} which implies a doubly stochastic $M$ has entries bounded by $2/\sqrt{d}$, otherwise there exists $X \succeq \mathbf{0}$ such that the ratio between $\iprod{ X, M - D + I}$ and $\trace X$ is large. In context of our lower bound on spectral sparsification error, this allows us to assume $H$ has edge weights bounded by $2 / \sqrt{d}$.

\smalledgewt*

\begin{proof}
Denote $\Pi^{\perp}$ by the orthogonal projector away from $S$, let $e_{a,b.\pm} = e_a \pm e_b$, and consider
\[ 
X = \Pi^{\perp} \cdot e_{ab+} e_{ab+}^{\T} \cdot \Pi^{\perp}
\qquad\text{and}\qquad 
Y = \Pi^{\perp} \cdot e_{ab-} e_{ab-}^{\T} \cdot \Pi^{\perp}
. 
\]
As $\trace X + \trace Y \le 4$, showing that 
\[ \iprod{X-Y,M-D+I} = \iprod{X,M-D+I} - \iprod{Y,M-D+I} \ge  \frac{4 \big( 2 - O(\sqrt{d}/n^{\delta}) \big)}{\sqrt{d}} = \frac{8}{\sqrt{d}} - O(1/n^{\delta})  \]
will suffice, as then by an arithmetic-mean--harmonic-mean inequality,
\[ \frac{\iprod{X,M-D+I}}{\trace X} - \frac{\iprod{Y,M-D+I}}{\trace Y} \ge \frac{4}{\sqrt{d}} - O(1/n^{\delta}) \]
and at least one of the two terms is at least $2/\sqrt{d} - O(1/n^{\delta})$ in absolute value. Note that 
\[ 2 \cdot \iprod{e_{a} e_{b}^{\T} + e_{b} e_{a}^{\T} , M-D+I} \ge 8/\sqrt{d},\]
so that it remains to show that, taking $\Delta = (X-Y) - 2(e_{a} e_{b}^{\T} + e_{b} e_{a}^{\T})$,
\[\iprod{\Delta, M-D+I} \ge - O(1/n^{\delta}). \]

Denote $r = \dim S$ and fix an orthonormal basis $v_1, \ldots, v_r$ for $S$. Writing $\Delta$ under this, we derive:
\begin{align*}
\Delta
&= \bigg( I - \sum_{i=1}^r v_i v_i^{\T} \bigg) \cdot 2(e_{a} e_{b}^{\T} + e_{b} e_{a}^{\T}) \cdot \bigg( I - \sum_{i=1}^r v_i v_i^{\T} \bigg) - 2(e_{a} e_{b}^{\T} + e_{b} e_{a}^{\T}) \\
&= - 2\sum_{i=1}^r \Big( v_i^{\T}e_{a}(v_i e_b^{\T} + e_bv_i^{\T}) +  v_i^{\T} e_{b} (v_i e_a^{\T} + e_av_i^{\T}) - \big( v_i^{\T} e_{a} \big)\big( v_i^{\T} e_{b} \big) v_i v_i^{\T} \Big).
\end{align*}
So,
\begin{align*}
&|\iprod{\Delta,M-D+I}| 
\\&\le \iprod{-2\sum_{i=1}^r \Big( v_i^{\T}e_{a}(v_i e_b^{\T} + e_bv_i^{\T}) +  v_i^{\T} e_{b} (v_i e_a^{\T} + e_av_i^{\T})-\big( v_i^{\T} e_{a} \big)\big( v_i^{\T} e_{b} \big) v_i v_i^{\T},M-D+I}
\\&\le 2\sum_{i=1}^r \left( \|v_i\|_{\infty} \cdot \Big( \big| 2v_i^{\T}(M-D+I) e_{a} \big| + \big| 2v_i^{\T} (M-D+I) e_b \big| \Big) + \|v_i\|_{\infty}^2 \cdot \big( v_i^{\T}(M-D+I)v_i \big) \right)
\\&\le 2\sum_{i=1}^r \left( \|v_i\|_{\infty} \cdot 4 \big( 1+4/\sqrt{d} \big) \cdot \|v_i\|_{\infty} + n \cdot \|v_i\|_{\infty}^4 \right)
\\&\le 8 \big( 1+4/\sqrt{d} \big) \cdot \frac{r}{n^{1/2+\delta}} + \big( 1+2/\sqrt{d} \big) \cdot \frac{r}{n^{2\delta}}
.
\end{align*}
and as $r \leq n^\delta$, we have $\iprod{ \Delta, M - D + I } \geq - O(1/n^{\delta})$ as required.
\end{proof}

\end{document}